\documentclass{article}

\usepackage[utf8]{inputenc}

\usepackage{xparse,xifthen}
\usepackage{mathtools} 
\usepackage{amsthm,amsxtra}
\usepackage{tikz,tikz-cd}
\usepackage[auth-sc,affil-sl]{authblk}
\usepackage{thmtools}
\usepackage{enumitem}
\usepackage{braket}
\usepackage{slashed}

\usepackage{mathpazo}
\usepackage{MnSymbol}
\usepackage{stmaryrd}

\usepackage{url}
\usepackage{hyperref}

\DeclareMathAlphabet{\mathsc}{OT1}{cmr}{m}{sc}
\makeatletter
\newcommand{\pushright}[1]{\ifmeasuring@#1\else\omit\hfill$\displaystyle#1$\fi\ignorespaces}
\newcommand{\pushleft}[1]{\ifmeasuring@#1\else\omit$\displaystyle#1$\hfill\fi\ignorespaces}
\makeatother

\newcommand{\lo}[1]{\raisebox{-0.1ex}{$#1$}\,}
\newcommand{\loo}[1]{\raisebox{-0.2ex}{$#1$}\,}
\newcommand{\Lo}[1]{\raisebox{-0.3ex}{$#1$}\,}
\newcommand{\Loo}[1]{\raisebox{-0.4ex}{$#1$}\,}
\newcommand{\LO}[1]{\raisebox{-0.5ex}{$#1$}\,}
\newcommand{\LOO}[1]{\raisebox{-0.6ex}{$#1$}\,}

\newcommand{\R}{\mathbb R}
\newcommand{\C}{\mathbb C}
\newcommand{\N}{\mathbb N}
\newcommand{\Z}{\mathbb Z}

\newcommand{\abs}[1]{\lvert #1 \rvert}
\newcommand{\Abs}[1]{\left\lvert #1 \right\rvert}
\newcommand{\norm}[1]{\lVert #1 \rVert}
\newcommand{\Norm}[1]{\left\lVert #1 \right\rVert}
\newcommand{\D}{\mathrm{d}}
\newcommand{\I}{\mathrm{i}}
\newcommand{\e}{\mathrm{e}}
\newcommand{\Langle}{\left\langle}
\newcommand{\Rangle}{\right\rangle}
\newcommand{\Cb}[1]{C_\mathrm{b}\!\left(#1\right)}
\newcommand{\stfirst}[2]{\mspace{1mu}\genfrac{[}{]}{0pt}{}{#1}{#2}\mspace{1mu}}
\newcommand{\stsecond}[2]{\mspace{0mu}\genfrac{\{}{\}}{0pt}{}{#1}{#2}\mspace{0mu}}

\DeclareDocumentCommand{\EE}{ O{}mO{} }{
	\ifthenelse{\isempty{#3}}{
		\mathbb E_{#1}\!\left[#2\right]}{
		\mathbb E_{#1}\!\left[#2\,\middle|\,#3\right]}
}

\newcommand{\mc}[1]{\mathcal{#1}}
\newcommand{\mb}[1]{\mathbb{#1}}
\newcommand{\mf}[1]{\mathfrak{#1}}
\newcommand{\mr}[1]{\mathrm{#1}}
\newcommand{\cat}[1]{\mathsc{#1}}

\DeclareMathOperator{\supp}{supp}

\DeclareMathOperator{\dom}{dom}

\declaretheorem[
	name=Definition, style=definition%
]{defn}
\declaretheorem[
	numbered=no, name=Hypothesis R, style=definition
]{hypo}

\declaretheorem[
	name=Proposition, style=plain, sibling=defn
]{prop}
\declaretheorem[
	name=Lemma, style=plain, sibling=defn
]{lem}
\declaretheorem[
	name=Theorem, style=plain, sibling=defn
]{theo}

\declaretheorem[
	name=Example, style=remark, sibling=defn
]{ex}
\declaretheorem[
	name=Remark, style=remark, sibling=defn
]{rk}

\author{%
Horst Thaler
}
\author{%
Rodrigo Vargas Le-Bert\thanks{Supported by Associazione LumbeLumbe.}
}
\affil{School of Architecture and Design, University of Camerino}
\date{\normalsize\today}

\title{Field Theory Done Right}

\begin{document}

\begin{titlepage}
\maketitle
\thispagestyle{empty}

\begin{abstract}
%
%
%
%
An effective formalism for white noise analysis, conceptually equivalent to Wilsonian renormalization theory, is introduced.
Space-time gets represented by a boolean lattice of coarse regions, energy scales become space-time partitions by lattice regions, and observables are elements of a projective limit with connecting maps 
given by partial integration of high-energy degrees of freedom.
The framework allows for a seamless generalization of the Wick product and the $\mc S$-transform to essentially arbitrary Lévy noises, and we provide a tool to make explicit calculations in several cases of interest, including Gauss, Poisson and Gamma noises (we shall thereby encounter pretty familiar polynomials, like falling factorials and Hermite polynomials).

Armed with this, we turn to constructive quantum field theory. We adopt an Euclidean approach and introduce a sufficient condition for reflection positivity, based on our $\mc S$-transform, enabling us to construct  non-trivial quantum fields by simply specifying compatible families of effective connected $n$-point functions. We exemplify this by producing a field with quartic interaction in dimension $d\leq 8$. Its connected $n$-point functions vanish except for the propagator and the connected $4$-point function, which is that of the $\phi^4$ field up to order $\hbar$. This model satisfies all the physical requirements of a non-trivial quantum field theory.
\medskip \\
{\bf 2010 MSC}:
60H40, 
81T08 
(primary);
81T16 
(secondary).
\end{abstract}

\end{titlepage}

\setcounter{tocdepth}{2}
\tableofcontents

\section{White noise analysis}

\subsection{Continuous product measures}

Let $M$ be a Riemannian manifold (space-time), over which we will study real-valued fields $x:M\rightarrow\R$. 
We approach our fields by considering their coarse-grained versions, obtained by taking local mean values.\footnote{From this effective perspective, any family of compatible coarse-grained fields could, a priori, be a valid field configuration and we accept them all---although we do expect that typical configurations of a given physically meaningful statistic ensemble can be taken to belong to a suitable space of not-so-general fields.}
In making precise sense of this it is natural to use the projection lattice of the von Neuman algebra $L^\infty(M,\C)$, and we start by recalling the relevant notions.

The space of measurable, essentially bounded functions $x:M\rightarrow\R$ modulo equality almost everywhere forms a (real) vector space $L^\infty(M) = L^\infty(M,\R)$, which becomes a Banach space once equipped with the essential supremum norm
\[
\norm{x}_\infty = \inf\Set{ \vphantom{\hat{\hat A}} C>0 | \abs{x(m)}\leq C\text{ for almost all } m\in M }.
\]
This Banach space is actually a Banach algebra for the pointwise product
\[
xy(m) = x(m)y(m),\quad x,y\in L^\infty(M)
\]
which is well-defined and satisfies $\norm{xy}_\infty \leq \norm x_\infty\norm y_\infty$\lo. A \emph{projection} $p\in L^\infty(M)$ is an element of this algebra satisfying $p^2=p$. It is plain to see that a projection can only take, essentially, the values 0 and 1, and is therefore (the equivalence class of) an indicator function
\[
p_A(m) = \left\{\begin{aligned} &1& &m\in A\\ &0& &\text{otherwise}\end{aligned}\right.
\]
for some measurable set $A\subseteq M$, which is well-defined modulo a set of measure 0 and can be taken to be the essential support of $p$. The set $\Lambda_\mr{meas} = \Lambda_\mr{meas}(M)$ of all projections $p\in L^\infty(M)$ forms a distributive lattice for the two operations
\[
p\wedge q = pq,\quad p\vee q = p+q-pq.
\]
A partial order relation on $\Lambda_\mr{meas}$ is then imposed by
\[
p\leq q\ \Leftrightarrow\ p=p\wedge q
\]
for $p,q\in \Lambda_\mr{meas}$\lo, which coincides with set inclusion of the corresponding essential supports, i.e. $p_A\leq p_B\Leftrightarrow A\subseteq B.$
There is also a least element $0=p_{\emptyset}$ and a greatest element $1=p_M$\lo.
Last but not least, each element $p\in \Lambda_\mr{meas}$ has a complement $\neg p=1-p$. Thus, it turns out that $\Lambda_\mr{meas}$ is even a Boolean lattice.

We proceed to define the space of generalized fields. Consider finite \emph{partitions} $P = \set{p_1\lo,\dots,p_m}\subseteq \Lambda_\mr{meas}$\lo. By this, we mean that we require \emph{completeness} and \emph{orthogonality,} in the sense that
\[
\sum p_i = 1,\quad p_ip_j=0 \text{ whenever } i\neq j.
\]
Write $X_P = L^\infty(M;P)\subseteq L^\infty(M)$ for the space of $P$-simple functions, i.e.\ linear combinations of projections in $P$.
Equip the family $\mc P_\mr{meas}$ of all such $P$'s with the partial order given by inclusion of the associated simple function subspaces, i.e.\
\[
Q\succcurlyeq P\ \Leftrightarrow\ L^\infty(M;Q)\supseteq L^\infty(M;P).
\]
If $Q\succcurlyeq P$, one has a projection
\[
\pi_{PQ}:X_Q\rightarrow X_P\lo,\quad
(\pi_{PQ}x)_p = \frac1{\abs p} \sum_{q\leq p} \abs{q}x_q\lo, 
\]
where $\abs p = \int_M p$ with respect to the volume measure. Next, choose a directed subset $\mc P\subseteq\mc P_\mr{meas}$\lo. The space of fields, whose topology will depend on a probability measure to be constructed and therefore cannot be specified yet, will be a subset of the (algebraic) projective limit $X = \projlim X_P$ taken over the partitions belonging to $\mc P$.
\begin{rk}
Let us elaborate on the convenience of allowing for the use a subfamily $\mc P\subseteq\mc P_\mr{meas}$\lo, as opposed to \emph{all} of $\mc P_\mr{meas}$\lo. The point is that the geometry of $M$ has a role to play in constructing physically relevant measures on $X$, while $\mc P_\mr{meas}$ encodes just its measure-theoretic structure (as far as $\mc P_\mr{meas}$ knows, $M$ is indistinguishable from either the interval $[0,1]$, if it is compact, or the real line $\R$, if it is not).
Given, as we will shortly see, that the algebra of local observables depends on $\mc P$, it might be desirable to have $\mc P$ reflect the geometry of $M$---specially if the geometric background is fixed, as it will be in all of the theories that we consider here.\footnote{Having the possibility of \emph{not} forcing $\mc P$ to encode the geometry of $M$ might of course be equally important, for instance in searching for models of quantum gravity.} We will shortly introduce one natural way of doing so, by relying on the smooth structure of $M$ as encoded in its possible piecewise smooth cellular decompositions.
\end{rk}

Now, take a convolution semigroup of probability measures $\Set{\nu_\lambda}_{\lambda\geq 0}$ on $\R$. We equip $X_P$ with the reference measure
\[
\D\mu_P(x) = \prod_{p\in P} \D\nu_{\abs p}(\abs px_p).
\]

\begin{prop}
Given $Q\succcurlyeq P$, one has $(\pi_{PQ})_*\mu_Q = \mu_P$\lo. In particular, the family $\mu = \set{\mu_P}$ defines a cylinder measure on $X$.
\end{prop}
\begin{proof}
Take an element $p\in P$ and write it as $p=\sum q_i$ with $q_i\in Q$. By independence, it suffices to check, assuming that $x_{q_i}$ has distribution $\D\nu_{\abs{q_i}}(\abs{q_i} x_{q_i})$, that $x_p = \frac1{\abs{p}} \sum \abs{q_i} x_{q_i}$ has distribution $\D\nu_{\abs{p}}(\abs px)$. And this follows from the fact that $\nu_\lambda$ is a convolution semigroup and $\abs p = \sum \abs{q_i}$.
\end{proof}
\begin{rk}
Let $\psi$ be the Lévy characteristic of $\nu_\lambda$\lo, i.e.\ $\hat\nu_\lambda(\xi) = \e^{\lambda\psi(\xi)}$.
One can convince oneself that, formally, the characteristic function of the measure $\mu$ is
\[
\EE{\e^{-\I\int_M\xi(m)x(m)\D m}} = \e^{\int_M \psi(\xi(m))\D m}.
\]
If instead of $X=\projlim X_P$ one takes a nuclear space of distributions on $M$, then $\mu$ could be constructed, as a Radon measure, by applying the Bochner-Minlos theorem to this characteristic function.
\end{rk}

\begin{defn}
Given a function $a\in L^1(X_Q) = L^1(X_Q\lo,\C)$, write $\EE{a}[P]\in L^1(X_P)$ for the conditional expectation of $a(x_Q)$ given $x_P=\pi_{PQ}(x_Q)$ with respect to the measure $\mu_Q$\lo.
A \emph{cylinder density} is a family $a = \set{a_P}$ of integrable functions on $X_P$ which satisfy the martingale condition $\EE{a_Q}[P] = a_P$\lo, so that $a\mu = \set{a_P\mu_P}$ is a (signed) cylinder measure on $X$. Thus, the space of cylinder densities is the (complex) vector space
\[
L^1_\mr{eff}(X) := \projlim L^1(X_P),
\]
where the projection $L^1(X_Q)\rightarrow L^1(X_P)$ for $Q\succcurlyeq P$ is $\EE{\cdot}[P]$. We will write $L^1(X)\subseteq L^1_\mr{eff}(X)$ for the subspace of cylinder densities satisfying
\[
\sup_{P\in\mc P_\mr{geom}} \Set{\Norm{a_P}_{L^1(X_P)}} \leq\infty.
\]
In general, however, this space will be too small to contain all the cylinder densities that we are interested in.
\end{defn}
\begin{rk}
The projective limit above is \emph{algebraic.} It will certainly be interesting to figure out the right topology for it, but we limit ourselves here to develop the purely algebraic aspects of the theory.
\end{rk}

Finally, we specify a partition family $\mc P_\mr{geom}\subseteq\mc P_\mr{meas}$ which seems a good choice in fixed-background situations. We consider projections associated to \emph{(piecewise smooth, regular) cellular structures} on $M$, by which we mean finite, graded partitions
\[
M = \bigcup_{\sigma\in C} \sigma,\quad  C = C^0\cup \cdots\cup C^d
\]
satisfying the following conditions:
\begin{enumerate}
	\item Each so-called \emph{$k$-cell} $\sigma\in C^k$ is homeomorphic to $\R^k$.
    \item The boundary $\partial\sigma$ of a cell $\sigma\in C^k$ belongs to the Boolean algebra generated by $C$, is piecewise smooth and is homeomorphic to $\mb S^{k-1} = \Set{ x\in\R^k | \abs x = 1}$.
\end{enumerate}
Given a cellular structure $C$, we get the partition $P_C = \Set{p_\sigma | \sigma\in C^d}$ where, we recall, $p_\sigma$ is (the equivalence class of) the indicator function of $\sigma\subseteq M$. We let $\mc P_\mr{geom} = \mc P_\mr{geom}(M)$ be the family of all such partitions.
\begin{rk} \label{P geom is directed}
This family is directed, because a smooth manifold admits a unique compatible piecewise linear structure, and a piecewise smooth, regular cellular structure is essentially a choice of piecewise linear chart.\footnote{For a short survey of piecewise linear topology, see~\cite[section 1]{benedetti2016smoothing}. Another brief source of useful information is~\cite[section 1]{forman1995discrete}.} Moreover, in refining a given cellular structure one can restrict oneself to considering (a chain of) single cell bisections~\cite{stallings1967lectures}.
\end{rk}

\subsection{Evaluation observables and their product}

The following rather simple property lies at the heart of our approach.
\begin{prop} \label{E(x_q)}
Given $Q\succcurlyeq P$ and $q\in Q$,
\(
\EE{x_q}[P] = x_p\lo,
\)
where $p\in P$ is uniquely determined by the condition $q\leq p$.
\end{prop}
\begin{proof}
A more general calculation will be made in \autoref{explicit-calc} under Hypothesis~R. The claim here follows from the proof of \autoref{cond_exp_of_powers} by the fact that $R^1(\lambda)$ exists and equals
\[
\begin{pmatrix} 1& 0\\ 0& 1/\lambda \end{pmatrix}
\]
even if Hypothesis~R does not hold.
\end{proof}

Let $\Lambda \subseteq\Lambda_\mr{meas}(M)$  be the sublattice generated by the projections in $\bigcup\mc P$. We will always think here of $\mc P$ as being $\mc P_\mr{geom}$\lo, but that makes no difference for the general theory.
$\Lambda$ inherits the order relation $\leq$ and the complementation operation $\neg p=1-p$ of $\Lambda_\mr{meas}$\lo, and is therefore a Boolean lattice too. In order to exploit \autoref{E(x_q)} to define field evaluation observables we need some elements of lattice theory, which we proceed to recall.

\begin{defn}
A \emph{filter} of a lattice $\Lambda$ (actually, just the partial order is required) is a set $\mf f\subseteq\Lambda$ which is:
\begin{enumerate}
	\item Nonempty and proper (i.e.\ not equal to all of $\Lambda$).
    \item Downward directed: given $p,q\in\mf f$, there is some $r\in\mf f$ with $r\leq p\wedge q$. By the next requirement, one can equivalently ask that $p\wedge q\in\mf f$.
    \item Upward saturated: if $q\leq p$ and $q\in\mf f$, then $p\in\mf f$.
\end{enumerate}
An \emph{ultrafilter} is a maximal filter. Equivalently, when $\Lambda$ is a Boolean lattice, an ultrafilter is a filter $\mf m$ such that, for every $p\in\Lambda$, either $p\in\mf m$ or $\neg p\in\mf m$---and it cannot be both, for then $\mf m$ would not be proper.
We will write $\mf M = \mf M(M,\Lambda)$ for the space of ultrafilters of $\Lambda$.
\end{defn}
\begin{prop}
Let $\mf m\in\mf M$ and $P\in \mc P$. The intersection $\mf m\cap P$ contains a unique projection, which will be written $p(\mf m)$.
\end{prop}
\begin{proof}
We show first that $\mf m\cap P$ is nonempty.
Assume that no element of the partition $P = \set{p_1\lo,\dots,p_n}$ belongs to $\mf m$. By maximality, $\neg p_i = 1-p_i\in\mf m$ and therefore
\[
\prod (1-p_i) = 1 - \sum p_i = 0 \in\mf m,
\]
contradicting properness. It remains to prove that $\mf m\cap P$ is a singleton; but if it had two elements $p_i\neq p_j$ it would also contain their product $p_ip_j=0$, again contradicting properness.
\end{proof}

Fix a projection $p_0\in\Lambda$ in a partition $P_0\in\mc P$.
By \autoref{E(x_q)}, every ultrafilter $\mf m$ of $\Lambda$ containing $p_0$ determines an \emph{ultraviolet completion} of the effective observable $a_{P_0}(x) = x_{p_0}$\loo, namely the cylinder density $\Set{a_P}$ defined by $a_P(x) = x_{p(\mf m)}$\loo. 

\begin{defn}
Given $\mf m\in\mf M$, we write $x(\mf m)\in L^1_\mr{eff}(X)$ for the cylinder density $\set{a_P}$ defined by $a_P(x) = x_{p(\mf m)}$\loo.
\end{defn}
\begin{rk}
It is easy to check that if $\nu_\lambda$ is Gaussian, then $x(\mf m)\nin L^1(X)$. If, on the other hand, $\nu_\lambda$ is supported on the positive reals, then
\[
\Norm{x(\mf m)}_{L^1(X_P)} = \EE{x(\mf m)}
\]
is independent of $P$ and $x(\mf m)\in L^1(X)$.
\end{rk}
\begin{rk}
It would be nice to have a correspondence between evaluation observables and points of $M$. That should be possible, at the price of introducing some extra geometric structure spoiling the invariance of our constructions under the symmetries of $M$. Indeed, instead of considering ``partitions'' of $M$ made of projections in $L^\infty(M)$ (thereby with support defined only up to a set of measure 0), one can take actual partitions, consisting of measurable subsets of $M$ with positive measure which belong to the algebra generated by a piecewise smooth cellular structure on $M$. In that framework, points of $M$ could be put in correspondence with ultrafilters of the Boolean algebra generated by the family of all such partitions. The problem is that this family is not directed (a common refinement of two partitions that differ only by sets of measure 0 must contain sets of measure 0); therefore, the choice of a directed subfamily is introducing extra structure.\footnote{In the case of $M=\R$ one could use, for instance, only partitions by half-open intervals which are closed from the right. This sort of arbitrariness is unavoidable.} This would later be an annoyance in establishing Euclidean invariance, for instance. Thus, we choose to let go of the idea of point evaluations, which at any rate is suspicious in the context of statistical field theory.
\end{rk}

Now take two ultrafilters $\mf m\neq\mf n$ and let $P\in\mc P$ be fine enough to distinguish them, meaning that $p(\mf m)\neq p(\mf n)$. Clearly, refinements $Q\succcurlyeq P$ also have $q(\mf m)\neq q(\mf n)$; therefore, by independence and \autoref{E(x_q)}, the compatibility condition is satisfied by the effective observables $a_Q(x) = x_{q(\mf m)}x_{q(\mf n)}$\loo. Assuming that $\nu_\lambda$ admits moments of all orders, for $Q\preccurlyeq P$ we can define
\[
a_Q(x) = \EE{x_{p(\mf m)}x_{p(\mf n)}}[Q],
\]
thus obtaining a cylinder density which we denote by $x(\mf m)\diamond x(\mf n)$. We extend this product to the case $\mf m=\mf n$ as follows.
\begin{defn}\label{wick-product}
Let $\mf m_1\lo,\dots \mf m_n\in \mf M$. The \emph{Wick product} $x(\mf m_1)\diamond\cdots\diamond x(\mf m_n)$ is defined by
\[
\bigl(x(\mf m_1)\diamond\cdots\diamond x(\mf m_n)\bigr)_P = \EE{x_{q_1}\cdots x_{q_n}}[P]
\]
where $Q\succcurlyeq P$ is fine enough to admit the existence of \emph{pairwise different} $q_i$'s such that $q_i\leq p(\mf m_i)$, for each $i$. 
\end{defn}
\begin{prop}\label{prop-wick-product}
The Wick product is well-defined, i.e.\ using the notation above, it does not depend on the choice of $Q$ and $q_i$'s.
\end{prop}
\begin{proof}
Assume that there is another family $R\succcurlyeq P$ providing \emph{pairwise different} $r_i$'s such that $r_i\leq p(\mf m_i)$.
We shall whenever necessary, identify a projection $p\in \Lambda_\mr{geom}$ with the cell $\sigma\in CM$ for which $p=p_\sigma$. First let us take an arbitrary $\tilde{q}_i$ obeying $\tilde{q}_i\leq q_i$ and define $\tilde{Q}$ as being a refinement of $Q$, which integrates $\tilde{q}_i$ into the cell complex $Q$. Then we have
\[
\EE{x_{q_1}\cdots x_{\tilde{q}_i}\cdots\ x_{q_n}}[P]=
\EE{\EE{x_{q_1}\cdots x_{\tilde{q}_i}\cdots\ x_{q_n}}[Q]}[P] =\EE{x_{q_1}\cdots x_{q_n}}[P]
\]
because $\EE{x_{q_1}\cdots x_{\tilde{q}_i}\cdots\ x_{q_n}}[Q]=x_{q_1}\cdots\hat{x}_{q_i}\cdots x_{q_n}\EE{x_{\tilde{q}_i}}[Q]=x_{q_1}\cdots x_{q_n}$, where $\hat{\;\cdot\;}$ means that the variable in question is omitted and the last equality is due to \autoref{E(x_q)}.
Since this can be done for each cell, we get
\begin{equation}\label{shrinking}
\EE{x_{\tilde{q}_1}\cdots \ x_{\tilde{q}_n}}[P] =\EE{x_{q_1}\cdots x_{q_n}}[P],
\end{equation}
for every family of $\tilde{q}_i$'s satisfying $\tilde{q}_i\leq q_i$.
Then choose $S$ to be a common refinement of $Q$ and $R$, fine enough to be able to pick in each of the cells $q_i$ and $r_i$ smaller representatives $\tilde{q}_i,\tilde{r}_i\in S$, which are surrounded only by empty cells, i.e. cells not supporting any $x$-variable. In the cell complex $S$ we can then move the $x_{\tilde{q}_i}$'s freely from one cell to any other cell, which is not already occupied by some $x_{\tilde{q}_j}$. More precisely, a movement from cell $s_1$ to a cell $s_2$ can be obtained by applying property (\ref{shrinking}) with $\tilde{q}=s_1,\,q=s_1\cup s_2$ for the move from $s_1$ to $s_1 \cup s_2$ and with $\tilde{q}=s_2,\,q=s_1\cup s_2$ for the move from $s_1 \cup s_2$ to $s_2$. By the movements just constructed we may finally achieve that $\tilde{q}_i =\tilde{r}_i$ for each $i$. Therefore
\[
\EE{x_{q_1}\cdots x_{q_n}}[P]=\EE{x_{\tilde{q}_1}\cdots x_{\tilde{q}_n}}[P]=
\EE{x_{\tilde{r}_1}\cdots x_{\tilde{r}_n}}[P]
=\EE{x_{r_1}\cdots x_{r_n}}[P]. \qedhere
\]
\end{proof}
Note that the Wick product is actually a collection of \emph{effective} Wick products which are compatible, namely
\[
x_{p_1}\diamond\dots\diamond x_{p_n} = (x(\mf m_1)\diamond\dots\diamond x(\mf m_n))_P\Lo,
\]
where $\mf m_i$ is an arbitrary ultraviolet completion of $p_i$\lo.


\begin{defn}
Given a scale $P$, let $\mc O_\mr{poly}(X_P)\subseteq L^1_\mr{eff}(X_P)$ be the algebra of effective Wick polynomials, namely the vector space of integrable functions on $X_P$ of the form
\[
x\mapsto \sum_{k=0}^n \sum_{p_1,\dots,p_k\in P} \alpha_{p_1\dots p_k}x_{p_1}\diamond\cdots\diamond x_{p_k} 
\]
equipped with the linear extension of Wick product.
\end{defn}


Now, while the effective observable algebra is generated by the $x_p$'s with the Wick product, its ``ultraviolet completion'' does not just consist of (complex) linear combinations of Wick products of field evaluations $x(\mf m)$, because one must allow for the possibility of varying the linear combination with the scale---thereby covering stochastic integrals, too. Indeed, consider a family $\alpha^n=\set{\alpha^n_P | P \in\mc P}$ of tensors
$\alpha^n_P = (\alpha_{p_1\dots p_n})\in\C^{P^n}$
satisfying the compatibility condition
\begin{equation} \label{kernel_compatibility_condition}
\alpha_{p_1\dots p_n} = \sum_{q_i\leq p_i} \alpha_{q_1\dots q_n}\LO.
\end{equation}
Let us formally write
\[
\int_{\mf M^n} x^{\diamondtimes n}\D\alpha^n = \int_{\mf m_1,\dots,\mf m_n} x(\mf m_1)\diamond\cdots\diamond x(\mf m_n)\D\alpha^n(\mf m_1\lo,\dots,\mf m_n)
\]
for the family of densities $\set{a_P}$ given by
\[
a_P(x) = \sum_{p_1,\dots,p_n} \alpha_{p_1\dots p_n} x_{p_1}\diamond\cdots\diamond x_{p_n}\Loo.
\]
\begin{prop}
Under the compatibility condition \eqref{kernel_compatibility_condition},
\(
\int_{\mf M^n} x^{\diamondtimes n}\D\alpha^n
\)
is a cylinder density.
\end{prop}
\begin{proof}
Indeed, given $P\preccurlyeq Q$ one has
\begin{align*}
\EE{a_Q}[P] &= \sum_{p_1,\dots,p_n}\sum_{q_i\leq p_i} \alpha_{q_1\dots q_n}\EE{ x_{q_1}\diamond\cdots\diamond x_{q_n}}[P] \\
	&= \sum_{p_1,\dots,p_n} \alpha_{p_1\dots p_n} x_{p_1}\diamond\cdots\diamond x_{p_n} = a_P\loo. \qedhere
\end{align*}
\end{proof}

\begin{defn}
Let $\mc O_\mr{poly}(X)$ be the algebra of polynomial cylinder densities (or polynomial \emph{chaos expansions}), namely the vector space of cylinder densities of the form
\[
\sum_{k=0}^n \int_{\mf M^k}x^{\diamondtimes k}\D\alpha^k
\]
for compatible families $\alpha^k=\set{\alpha^k_P}$ of tensors $\alpha^k_P\in\C^{P^k}$. Note that this is indeed an algebra: given a second set of compatible families $\beta^\ell=\set{\beta^\ell_P}$, one has that $\alpha^k\otimes\beta^\ell = \set{\alpha^k_P\otimes \beta^\ell_P}$ satisfies the compatibility condition, too. Therefore, the Wick product
\[
\left( \sum_k \int_{\mf M^k}x^{\diamondtimes k}\D\alpha^k
 \right)\diamond \left( \sum_\ell \int_{\mf M^\ell} x^{\diamondtimes \ell}\D\beta^\ell \right) = \sum_{k,\ell} \int_{\mf M^{k+\ell}} x^{\diamondtimes(k+\ell)} \D(\alpha^k\otimes\beta^\ell)
\]
is well-defined.
\end{defn}
\begin{ex}
Given a bounded operator $A:L^2(M)\rightarrow L^2(M)$, it is straightforward to check that the coefficients
\[
\alpha_{p_1p_2} = \Langle p_1, Ap_2\Rangle,\quad p_1\lo,p_2\in\Lambda
\]
satisfy the compatibility condition \eqref{kernel_compatibility_condition}, and therefore define a polynomial observable $a\in\mc O_\mr{poly}(X)$.
\end{ex}

\subsection{The $\mc S$ transform and Wick calculus}

Next, we study the Wick product via Fourier transform, much in the spirit of classical white noise analysis~\cite{kondratiev1996wick, janson1997gaussian}. 
This will enable us to find an enlargement $\mc O(X)\supseteq \mc O_\mr{poly}(X)$ of the algebra of local polynomial observables which supports a much more flexible functional calculus.

Identify $X_P^*\cong X_P$ using the pairing
\[
\xi x = \sum_{p\in P} \xi_px_p\abs p
\]
and define, given $a_P\in L^1(X_P)$,
\(
(\mc Ta_P)(\xi) = \EE{\e^{-\I\xi x} a_P}.
\)
Now take a cylinder density $a = \Set{a_P}\in L^1_\mr{eff}(X)$. We want to compare $\mc Ta_P$ and $\mc Ta_Q$\Lo. In order to do so, consider the inclusion $\iota_{QP} = \pi_{PQ}^*: X_P^*\rightarrow X_Q^*$\LO, explicitly given by
\[
(\iota_{QP}\xi)_q = \xi_p\lo,
\]
where $p\geq q$.
\begin{prop}
The map
\[
\mc T: a\in L^1_\mr{eff}(X) \mapsto \EE{\e^{-\I\xi x}a} = \Set{\mc Ta_P} \in C(X^*) := \projlim C(X_P^*)
\]
is well-defined. In particular, $\mc T$ is well-defined on $\mc O_\mr{poly}(X)$.
\end{prop}
\begin{proof}
Indeed, given $\xi\in X_P^*$\loo,
\[
(\mc Ta_Q)(\iota_{QP}\xi)
	= \EE{\EE{ \e^{-\I\xi(\pi_{PQ}x)}a_Q }[P] }
    = \EE{ \e^{-\I\xi x} a_P }
    = (\mc Ta_P)(\xi). \qedhere
\] 
\end{proof}
\begin{rk}
A priori, $\mc Ta$ might well be unbounded: we only have boundedness on each $X_P^*$ and the bounds might not be uniform.
\end{rk}

\begin{prop}
Let $\hat\mu(\xi) = \EE{\e^{-\I\xi x}}$. One has that
\[
\mc T\bigl(x(\mf m)^{\diamond n}\bigr) = \hat\mu^{1-n} \mc T\bigl(x(\mf m)\bigr)^n.
\]
Thus, defining
\(
\mc Sa = \hat\mu^{-1}\mc Ta,
\)
one has $\mc S(a\diamond b) = \mc S(a)\mc S(b)$, for all $a,b\in\mc O_\mr{poly}(X)$.
\end{prop}
\begin{proof}
Assuming the first claim, one sees that
\[
\mc S\bigl(x(\mf m)^{\diamond(n+m)}\bigr) = \mc S\bigl( x(\mf m)^{\diamond n}\bigr) \mc S\bigl( x(\mf m)^{\diamond m}\bigr).
\]
The same holds for the Wick product of two general monomials $x(\mf m_1)^{\diamond n_1}\diamond\cdots\diamond x(\mf m_k)^{\diamond n_k}$ by independence, and by bilinearity for the Wick product of fully general Wick polynomials.

As for the first claim, let us do the case $n=2$. By definition,
\[
\EE{\e^{-\I\xi x} x_{p_0}^{\diamond 2}} = \EE{\e^{-\I\xi x} x_{q_0}x_{q_1}}
\]
where $q_0\neq q_1$ are both contained in $p_0$\lo.
Now, the right hand side equals
\[
\EE{\e^{-\I\sum_{q\neq q_1}\xi_q x_q\abs q}x_{q_0}} \EE{\e^{-\I\xi_{q_1}x_{q_1} \abs{q_1}} x_{q_1}} = \frac{\EE{\e^{-\I\xi x}x_{q_0}}}{\EE{\e^{-\I\xi_{q_1} x_{q_1}\abs{q_1}}}} \frac{\EE{\e^{-\I\xi x} x_{q_1}}}{\EE{\e^{-\I\sum_{q\neq q_1} \xi_qx_q\abs q}}}
\]
and the desired result follows. The case $n>2$ is done in the same way.
\end{proof}


\begin{defn}\label{fourier-stieltjes}
Write $B(X_P^*)$ for the Fourier-Stieltjes algebra of $X_P^*$\lo, i.e.\ the  subalgebra of $C(X_P^*)$ formed by Fourier transforms of  complex Radon measures on $X_P$\loo. Let
\[ 
\quad B_\mu(X_P^*)
	= \Bigl\{\varphi\in C(X_P^*)\,\Bigm|\, \varphi\hat\mu_P^\lambda\in B(X_P^*) \text{ and } (\varphi\hat\mu_P^\lambda)\raisebox{0.2ex}{$\check{\,}$}\ll \mu_P\lo,\ \text{ for all } \lambda>0  
\Bigr\}.
\] 
Note that the inverse of the $\mc S$-transform is well-defined on $B_\mu(X_P^*)$ and define
\[
\mc O(X_P) = \mc S^{-1}B_\mu(X_P^*) \subseteq L^1(X_P).
\]
Finally, let $\mc O(X) \subseteq L^1_\mr{eff}(X)$ be the space of cylinder densities with $a_P\in\mc O(X_P)$.
\end{defn}
\begin{rk}
Under mild conditions on $\nu$, $B_\mu(X_P^*)$ is a subalgebra of $C(X_P^*)$  and
$\mc O_\mr{poly}(X)\subseteq \mc O(X)$, see \autoref{observable algebra}.
\end{rk}

\begin{defn}
Let $a_1\lo,\dots,a_n \in\mc O(X)$, say $a_i = \set{a_{iP}}$ with $a_{iP}\in\mc O(X_P)$. Write $\varphi_{iP} = \mc Sa_{iP}\in B_\mu(X_P^*)$, and suppose that $f\in C(\R^n)$ is such that
\[
f(\varphi_{1P}\lo,\dots,\varphi_{nP}):\xi\in X_P^*\mapsto f\bigl( \varphi_{1P}(\xi),\dots,\varphi_{nP}(\xi) \bigr)
\]
belongs to $B_\mu(X_P^*)$, for all $P\in\mc P$. Note that the compatibility condition
\[
f(\varphi_{1P}\lo,\dots,\varphi_{nP})(\xi) = f(\varphi_{1Q}\lo,\dots,\varphi_{nQ})(\iota_{QP}\xi), \quad \xi\in X_P^*
\]
holds trivially.
Thus, we can define
\[
f^\diamond(a_1\lo,\dots,a_n) = \mc S^{-1} f(\mc Sa_1\lo,\dots,\mc Sa_n),
\]
where $f(\mc Sa_1\lo,\dots,\mc Sa_n) = \Set{f(\mc Sa_{1P}\lo,\dots,\mc Sa_{nP})}\in C(X^*)$.
\end{defn}
\begin{rk}
In particular, the Wick product is well-defined on $\mc O(X)$ and satisfies
\(
\mc S(a\diamond b) = \mc S(a)\mc S(b).
\)
\end{rk}

In order to check whether a function $f\in C(\R^n)$ satisfies the conditions required for the definition of $f^\diamond(a_1\lo,\dots,a_n)$ it is useful to know that the Banach-Stieltjes algebra of $X_P^*$ is spanned by the space of positive definite functions. More precisely, one has the following characterization of cylinder probability measures.

\begin{prop}
Let $\varphi\in C(X^*)$ be a \emph{positive definite} function, in the sense that
\[
\sum \bar c_kc_\ell\varphi_P(\xi_\ell-\xi_k) \geq 0
\]
for all $\xi_1\lo,\dots,\xi_n\in X_P^*$ and $c_1\lo,\dots,c_n\in\C$, such that $\varphi_P(0)=1$. Then, $\varphi$ is the characteristic function of a cylinder probability measure $\set{\mu_P}$ on $X$.
\end{prop}
\begin{proof}
Indeed, by Bochner's theorem, $\varphi_P$ is the characteristic function of a probability measure $\mu_P$ on $X_P$. The compatibility condition on trigonometric polynomials $a_P(x) = \sum c_k\e^{-\I\xi_k x}$ is directly verified:
\begin{align*}
\EE{a_P}
	&= \EE{\sum c_k\e^{-\I\xi_k x}}
	= \sum c_k\varphi_P(\xi_k) = \sum c_k\varphi_Q(\iota_{QP}\xi_k) \\
	&= \EE{\sum c_k\e^{-\I\xi_k\pi_{PQ}x}}
    = \EE{\pi_{PQ}^*a_P}.
\end{align*}
We conclude by density of the trigonometric polynomials in $\Cb{X_P}$.
\end{proof}

\subsection{Fields whose Wick polynomials are polynomials}\label{explicit-calc}

Let $\C[X_P]$ be the algebra of polynomials on the variables $\set{x_p | p\in P}$. In order to do explicit calculations it is desirable that the reference measure $\mu$ is such that the expected value of a polynomial $f\in\C[X_Q]$ of degree $n$ is a polynomial $\EE{f}[P]\in\C[X_P]$ of degree $n$, too. This can be ensured by making the following assumption.
\begin{hypo}
For each $\lambda\geq 0$ and $k\in\N$, there are constants $\set{c_{k\ell}(\lambda) | \ell\leq k+1}$ such that
\[\label{hypothesis-R}
\hat \nu_\lambda^{(k)}\hat \nu_\lambda' = \sum_{\ell=0}^{k+1} c_{k\ell}(\lambda) \hat\nu_\lambda^{(\ell)}\hat\nu_\lambda\lo.
\]
Under this assumption, it can be seen~\cite{vargas2017wick} that the vector spaces generated by the sets $\Set{\hat\nu_\lambda^{(\ell)} | \ell\leq k}$ and $\Set{\hat\nu^{(\ell)} \hat\nu_{\lambda-1} | \ell\leq k}$ are equal, so that we can define the matrix $R^k(\lambda)\in\mathrm{GL}(k+1,\C)$ by the equation
\begin{equation}\label{R-matrix}
\hat\nu_{\lambda-1} \left(\begin{matrix} \hat\nu \\ \hat\nu' \\ \vdots \\ \hat\nu^{(k)} \end{matrix}\right)
	= R^k(\lambda) \left(\begin{matrix} \hat\nu_\lambda \\ \hat\nu_\lambda' \\ \vdots \\ \hat\nu_\lambda^{(k)} \end{matrix}\right).
\end{equation}
In particular, $R^1(\lambda) = \begin{pmatrix} 1 &0 \\ 0 &1/\lambda \end{pmatrix}$, but further terms will depend on $\nu$.
Observe that we can safely drop the superscript $k$, because $R^k(\lambda)$ is a lower-triangular matrix obtained from $R^{k+1}(\lambda)$ by simply erasing the last line and column---and the same will apply to their inverses.
We also define $R(\lambda_1\lo,\lambda_2) = R(\lambda_1)^{-1}R(\lambda_2)$, so that
\begin{equation}\label{R-matrix-2}
\underbrace{\hat\nu_{1-\lambda_1} \hat\nu_{\lambda_2-1}}_{\hat\nu_{\lambda_2-\lambda_1}} \begin{pmatrix} \hat\nu_{\lambda_1} \\ \hat\nu_{\lambda_1}' \\ \hat\nu_{\lambda_1}'' \\ \vdots  \end{pmatrix} = R(\lambda_1\lo,\lambda_2) \begin{pmatrix} \hat\nu_{\lambda_2} \\ \hat\nu_{\lambda_2}' \\ \hat\nu_{\lambda_2}'' \\ \vdots \end{pmatrix}
\end{equation}
and $R(\lambda_1\lo,\lambda_2)R(\lambda_2\lo,\lambda_3) = R(\lambda_1\lo,\lambda_3)$.
\end{hypo}

In order to simplify calculations in what follows, suppose that the support of $\nu_\lambda$ is a fixed additive semigroup $S\subseteq \R$ (in the examples that we will consider, $S=\R,\R_+\lo,\Z$ or $\N$)
and that $\D\nu_\lambda(s) = \rho_\lambda(s)\D s$ where $\D s$ is an invariant measure (either Lebesgue or counting measure, whichever is appropriate). Thus,
\[
\int_\R f(x)\D\nu_\lambda(\lambda x) = \int_S f(s/\lambda)\rho_\lambda(s)\D s.
\]
\begin{lem}\label{cond_exp_lem}
Assume Hypothesis R. Convening that the entries of $R$ are indexed starting from 0, one has that
\begin{align*}
& \int t_2^k \mspace{1mu} \rho_{\mu_2}(t_2) \rho_{\lambda-\mu_1-\mu_2}(\lambda s - t_1-t_2 )\D t_2 \\
&\quad = \rho_\lambda(\lambda s-t_1) \sum_{\ell=0}^k (-\I)^k R_{k\ell}(\mu_2,\lambda-\mu_1)\bigl(\I(\lambda s-t_1)\bigr)^\ell.
\end{align*}
In particular,
\( 
\int t^k \rho_\mu(t) \rho_{\lambda-\mu}(\lambda s - t )\D t = \rho_\lambda(\lambda s) \sum_{\ell=0}^k (-\I)^k R_{k\ell}(\mu,\lambda)(\I \lambda s)^\ell.
\) 
\end{lem}
\begin{proof}
Indeed, $\int t_2^k \mspace{1mu}\rho_{\mu_2}(y_2) \rho_{\lambda-\mu_1-\mu_2}(\lambda s - t_1-t_2 )\D t_2$ equals
\begin{align*}
&\left((-\I)^k\hat\nu_{\mu_2}^{(k)} \hat\nu_{\lambda-\mu_1-\mu_2} \right)\raisebox{0.8ex}{$\check{}$}\, (\lambda s-t_1) \\
&\quad = (-\I)^k \sum_{\ell=0}^k R_{k\ell}(\mu_2\lo,\lambda-\mu_1) \bigl(\I (\lambda x-y_1)\bigr)^\ell \rho_{\lambda}(\lambda s-t_1). \qedhere
\end{align*}
\end{proof}

\begin{theo} \label{cond_exp_of_powers}
Assuming Hypothesis R,
let $q_1\lo,q_2\in Q$ and $P\preccurlyeq Q$ be such that $p:=q_1+q_2\in P$. One has that
\begin{align*}
&\EE{x_{q_1}^{k_1}x_{q_2}^{k_2}}[P] \\ &\quad= 
\left(\frac{\abs{q_1}}{\abs{q_2}}\right)^{k_2} \sum_{j=0}^{k_2} (-1)^j {k_2\choose j} \sum_{\ell=0}^{k_1+j} \bigl(\I\abs{q_1}\bigr)^{-k_1-k_2} \bigl(\I\abs p\bigr)^{k_2+\ell-j} R_{k_1+j,\ell}\bigl( \abs{q_1},\abs p \bigr) x_p^{k_2+\ell-j}\Lo.
\end{align*}
\end{theo}
\begin{proof}
By independence, we can work locally, i.e.\ on the projection lattice of $pL^\infty(M)$; thus, we assume that $p=1$ and $P=\set{1}$. Now, write $s=\abs px_p\in S$ and $t = \abs{q_1}x_{q_1}\in S$, so that $\abs{q_2}x_{q_2} = s-t$. We have that 
\[
\D\nu_{\lambda_1}\bigl(\abs{q_1}x_{q_1}\bigr) \D\nu_{\lambda_2}\bigl(\abs{q_2}x_{q_2}\bigr)
  = \rho_{\lambda_1}(t) \rho_{\lambda_2}(s-t) \D s\D t.
\]
In terms of these variables 
and applying \autoref{cond_exp_lem},
$\EE{x_{q_1}^{k_1}x_{q_2}^{k_2}}[P]$  equals
\begin{align*}
&\frac1{\rho_{\abs p}(s)} \int_S \left(\frac{t}{\abs{q_1}}\right)^{k_1} \left( \frac{s-t}{\abs{q_2}} \right)^{k_2} \rho_{\abs{q_1}}(t) \rho_{\abs{q_2}}(s-t) \D t \\
&\quad= \frac1{\abs{q_1}^{k_1}\abs{q_2}^{k_2}} \sum_{j=0}^{k_2}{k_2\choose j} \bigl(\abs px_p\bigr)^{k_2-j} (-1)^j \sum_{\ell=0}^{k_1+j} \bigl(-\I\bigr)^{k_1+j} R_{k_1+j,\ell}\bigl(\abs{q_1}\lo,\abs p\bigr)\bigl(\I\abs px_p\bigr)^\ell \\
    &\quad= \left(\frac{\abs{q_1}}{\abs{q_2}}\right)^{k_2} \sum_{j=0}^{k_2} (-1)^j  {k_2\choose j} \sum_{\ell=0}^{k_1+j} \bigl(\I\abs{q_1}\bigr)^{-k_1-k_2} \bigl(\I\abs p\bigr)^{k_2+\ell-j} R_{k_1+j,\ell}\bigl( \abs{q_1},\abs p \bigr) x_p^{k_2+\ell-j}\loo,
\end{align*}
as claimed.
\end{proof}
\begin{rk}
Iterating, one can obtain explicit formulas for $\EE{x_{q_1}^{k_1}\cdots x_{q_n}^{k_n}}[P]$ for arbitrary $P\preccurlyeq Q$.
\end{rk}

\section{Wick polynomial calculations}

In this section we want to see by means of concrete examples what the Wick product, as introduced in Definition \ref{wick-product}, amounts to. The calculations will make it evident that for Gamma noise this Wick product is the same as multiplicative renormalization, whereas for Poisson and Gauss noises it encodes an additive renormalization and the Wick products themselves are given by appropriately scaled falling factorials and Hermite polynomials, respectively.

\subsection{$\Gamma$ noise}

The calculation of expected values of monomials for a $\Gamma$ reference measure, i.e.\
\[
\hat\rho(\xi) = (1+\I\xi)^{-1} 
\]
were computed in~\cite{vargas2017wick}. We recall the results here, in order to emphasize that renormalization can also be multiplicative, as opposed to purely additive.

The $\Gamma$ field is particularily simple. It satisfies Hypothesis R and the $R$ matrix turns out to be diagonal:
\[
R(\lambda) = \begin{pmatrix}
	1 & & & \\
     &\lambda & & \\
     & &\frac{\lambda^{(2)}}{2!} & \\
     & & &\ddots
\end{pmatrix}
\]
where $\lambda^{(n)}$ is the rising factorial
\[
\lambda^{(n)} = \lambda(\lambda+1)\cdots(\lambda + n - 1).
\]
This enables one to do explicit calculations directly, and one finds that
\[
\EE{x_{q_1}^{k_1}\cdots x_{q_n}^{k_n}}[P] = \frac{\abs p^k}{\abs p^{(k)}}\left( \prod_{i=1}^k \frac{\abs{q_i}^{(k_i)}}{\abs{q_i}^{k_i}} \right) x_p^k\loo,\quad k=k_1+\cdots+k_n
\lo,
\]
which diverges as $Q\succcurlyeq P$ gets finer (so that each $\abs{q_i}$ becomes vanishingly small) as soon as some $k_i>1$.
However, it is clear that upon defining $x_{q,\mr{ren}}^k := C_q x_q^k$ with $C_q = \frac{\abs q^k}{\abs q^{(k)}}$ one has
\[
\EE{x_{q,\mr{ren}}^k}[P] = x_{p,\mr{ren}}^k\Lo.
\]
Thus, the factor $C_q$ is taking care of the divergence, while at the same time turning the family $\Set{x_{p(\mf m),\mr{ren}}^k}$ into a cylinder density. In other words, $\Set{x_{p(\mf m),\mr{ren}}^k}$ can play the role of the Wick power $x(\mf m)^{\diamond k}$. From this point of view one is actually recovering the result of applying \autoref{wick-product}, namely
\[
\bigl(x(\mf m)^{\diamond k}\bigr)_P = \EE{x_{q_1}\cdots x_{q_k}}[P] = x_{p(\mf m),\mr{ren}}^k
\]
where $q_1\lo,\dots, q_k\in Q\succcurlyeq P$ are pairwise different projections with $q_i\leq p$.

\subsection{Poisson noise} \label{Poisson}

Let $\hat\rho(\xi)=\e^{\alpha(\e^{-\I\xi}-1)}$ be the Fourier transform of the Poisson mass distribution $\rho(s)=\e^{-\alpha}\alpha^s/s!,\,s\in\N,\,\alpha\in \R_{>0}$\lo.
As before, $\hat\rho_\lambda:=\hat\rho^\lambda$ will denote the corresponding semigroup. By a standard induction argument one can show that
\[
\hat\rho_\lambda^{(k)} = (-\I)^k\sum_{\ell=0}^k {k\brace\ell} \left(\alpha\lambda\e^{-\I\xi}\right)^\ell \hat\rho_\lambda\lo,\quad k\in \N,
\]
where ${k\brace\ell}$ denotes Stirling numbers of the second kind.
Hypothesis R can be verified easily. First, calculate
\[
\hat\rho_\lambda^{(k)}\hat{\rho}'_\lambda=(-\I)^{k+1}\sum_{m=0}^{k+1}\stsecond{k}{m-1}(\alpha\lambda)^m \e^{-\I m\xi}\hat{\rho}^2_\lambda
\]
and
\[\sum_{\ell=0}^{k+1}c_{k\ell}(\lambda)\hat{\rho}^{(\ell)}_\lambda\hat{\rho}_\lambda =
\sum_{m=0}^{k+1}\left(\sum_{\ell=m}^{k+1}c_{k\ell}(\lambda)\I^{-\ell}\stsecond{\ell}{m}(\alpha\lambda)^m\right)\e^{-\I m\xi}\hat{\rho}^2_\lambda\LO.
\]
Then, comparing like coefficients for $\e^{-\I m\xi}$, gives the system of equations
\begin{equation}
\stsecond{k}{m-1}=\sum_{\ell=m}^{k+1}c_{k\ell}(\lambda)\I^{k-\ell+1}\stsecond{\ell}{m},\quad 0\leq m\leq k+1,
\end{equation}
which has to be solved in terms of the unknowns $c_{k\ell}(\lambda),\,0\leq \ell\leq k+1$.
But the square matrix $A=\left(\I^{k-\ell+1}\stsecond{\ell}{m}\right)$, $0\leq m$, $\ell\leq k+1$, has full rank, because it is an upper triangular matrix with $\mathrm{diag}(A)=\left(\I^{k-m+1}\right)_{0\leq m\leq k+1}$\LOO. The system is therefore solved by a unique vector of $c_{k\ell}(\lambda)$'s.

Let us compute the $R$ matrix.
Since $R(\lambda,\mu)$ is a two-parameter semigroup, we can as well focus on its generator.
Recall also that $R(\lambda)=R(1,\lambda)$ is defined by $u(1) = R(\lambda)u(\lambda)$ where
\begin{equation} \label{u(lambda)}
u(\lambda) = \begin{pmatrix} \hat\rho \\ (\hat\rho^\lambda)'\hat\rho^{1-\lambda}\\ \vdots \\ (\hat\rho^\lambda)^{(k)}\hat\rho^{1-\lambda} \end{pmatrix}.
\end{equation}
We want to find out the differential equation that $u$ obeys.

\begin{lem} \label{Stirling sum identity}
One has that
\[
\ell{k\brace \ell} = \sum_{j=\ell}^k {k\choose j-1} (-1)^{k-j} {j\brace \ell}.
\]
\end{lem}
\begin{proof}
We found this identity by working out the first few cases and verified it using Manuel Kauers' Mathematica package\footnote{\url{http://www.kauers.de/software.html}.} ``Stirling''~\cite{kauers2007summation}. 
Later, we posted it as a question on MathOverflow and got two nice answers. We reproduce~\cite{287808} here for convenience, but see also~\cite{287807}.

The identity can be interpreted as an instance of inclusion-exclusion. The left hand side counts the number of ways of partitioning $S=\set{1,2,\dots,k}$ into $\ell$ parts and then picking one of the parts as the designated one.
Let $A_i$ denote the set of partitions of $S$ into $\ell$ parts where the designated part contains $i$. It is plain to see that the left hand side is counting $\Abs{ A_1\cup A_2\cup\cdots\cup A_k }$. For the right hand side notice that
\[
\Abs{ A_{i_1}\cap A_{i_2}\cap\cdots\cap A_{i_r} } = {k-r-1\brace \ell}.
\]
So by inclusion-exclusion we get
\[
\Abs{A_1\cup A_2\cup\cdots\cup A_k} = \sum_{r=1}^{k-\ell+1} (-1)^{r-1} {k\choose r} {k-r+1\brace\ell}
\]
and reindexing by $j=k-r+1$ gives the desired identity.
\end{proof}

\begin{prop}\label{generator-poisson}
If $u=u(\lambda)$ is defined by \eqref{u(lambda)}, then $\frac{\D u}{\D\lambda} = \frac1\lambda Au$, where
\[
A = (A_{kj}),\quad A_{kj} = \left\{ \begin{aligned} & \I^{k-j} {k\choose j-1}& &j\leq k, \\ &0& &\text{otherwise.} \end{aligned}\right.
\]
Thus, $u(\lambda) = \e^{\log(\lambda)A}u(1)$.
\end{prop}
\begin{proof}
We compute
\begin{equation} \label{derivative}
\frac\D{\D\lambda} (\hat\rho^\lambda)^{(k)} \hat\rho^{1-\lambda} = \left( \log(\hat\rho) \hat\rho^\lambda \right)^{(k)} \hat\rho^{1-\lambda} - (\hat\rho^\lambda)^{(k)} \log(\hat\rho)\hat\rho^{1-\lambda}.
\end{equation}
Since
\(
\log\hat\rho = \alpha(\e^{-\I\xi}-1)
\)
and $(\hat\rho^\lambda)' = -\I\lambda \alpha\e^{-\I\xi} \hat\rho$, we get
\[
\left( \log(\hat\rho)\rho^\lambda \right)^{(k)} = \left( \frac{\I}{\lambda}(\hat\rho^\lambda)' - \alpha\hat\rho^\lambda \right)^{(k)} = \frac{\I}{\lambda}(\hat\rho^\lambda)^{(k+1)} - \alpha(\hat\rho^\lambda)^{(k)}.
\]
On the other hand,
\begin{align*}
\log(\hat\rho) (\hat\rho^\lambda)^{(k)}
	&= \alpha(\e^{-\I\xi}-1)(-\I)^k\sum_{\ell=0}^k {k\brace\ell} (\alpha\lambda\e^{-\I\xi})^\ell\hat\rho^\lambda \\
	&= \frac{\I}{\lambda}(-\I)^{k+1} \sum_{\ell=0}^k {k\brace\ell} (\alpha\lambda\e^{-\I\xi})^{\ell+1} \hat\rho^\lambda - \alpha(\rho^\lambda)^{(k)} \\
	&= \frac{\I}{\lambda}(-\I)^{k+1}\sum_{\ell=1}^{k+1} {k\brace\ell-1} (\alpha\lambda\e^{-\I\xi})^\ell \hat\rho^\lambda - \alpha(\rho^\lambda)^{(k)} \\
	&= \frac{\I}{\lambda}(-\I)^{k+1}\sum_{\ell=0}^{k+1} \left({k+1\brace\ell} - \ell{k\brace\ell} \right) (\alpha\lambda\e^{-\I\xi})^\ell \hat\rho^\lambda - \alpha(\rho^\lambda)^{(k)}
\end{align*}
and therefore, using \autoref{Stirling sum identity},
\begin{align*}
\frac\D{\D\lambda} (\hat\rho^\lambda)^{(k)} \hat\rho^{1-\lambda}
	&= \frac1{\lambda} (-\I)^{k} \sum_{\ell=0}^{k} \ell{k\brace\ell} (\alpha\lambda\e^{-\I\xi})^\ell \hat\rho \\
	&= \frac1{\lambda} (-\I)^{k} \sum_{\ell=0}^{k}  \sum_{j=\ell}^k {k\choose j-1} (-1)^{k-j} {j\brace \ell} (\alpha\lambda\e^{-\I\xi})^\ell \hat\rho \\
	&= \frac1{\lambda}  \sum_{j=0}^{k} (-\I)^{k-j}  {k\choose j-1} (-1)^{k-j} (-\I)^j \sum_{\ell=0}^j {j\brace \ell} (\alpha\lambda\e^{-\I\xi})^\ell \hat\rho \\
	&= \frac1{\lambda}  \sum_{j=0}^{k} (-\I)^{k-j}  {k\choose j-1} (-1)^{k-j} (\hat\rho^\lambda)^{(j)}\hat\rho^{1-\lambda}. \qedhere
\end{align*}
\end{proof}
\begin{lem}
\label{lemma-poisson}
The generator $A$ of Proposition \ref{generator-poisson} is diagonalizable with $A=UDU^{-1}$, where $D=\mathrm{diag}(0,\ldots,n)$ and $U\equiv (U_{kj})$ is given by $U_{kj}=\I^{k-n} {k \brace j}$, $n\in \mathbb{N}$ and $0\leq k,j \leq n$.
\end{lem}
\begin{proof}
Let $n\in\mathbb{N}$ be arbitrary but fixed. First we note that with respect to the basis $\mathcal{B}_1:=(1,x,\ldots,x^n)$ the operator $x\frac{d}{dx}$ is simply given by the diagonal matrix $D=\mathrm{diag}(0,\ldots,n).$ Then let us introduce a new basis $\mathcal{B}_2:=(\psi_0,\ldots,\psi_n),$
where
\begin{equation}\label{basis2}
\psi_k := \I^{k-n}\sum_{j=0}^k {k \brace j} x^j = \sum_{j=0}^k U_{kj}x^j = \left(U\mathcal{B}_1^T\right)_k =: \I^{k-n}\phi_k\lo.
\end{equation}
The $\phi_k$'s so introduced are also called exponential or Touchard polynomials. Now, the following relations hold, see e.g.\ \cite[Ch.\ 4, {S}ection 1.3]{roman2005umbral}:
\begin{equation}
\phi_{k+1}(x)=\left(x+x\frac{d}{dx}\right)\phi_k(x),
\end{equation}
and
\begin{equation}
-x\phi_k(x)=\sum_{j=0}^k {k \choose j-1}(-1)^{k-j}\phi_j-\phi_{k+1}\lo.
\end{equation}
Then, because of $x\frac{d}{dx}=-x+\left(x+x\frac{d}{dx}\right)$, we have $x\frac{d}{dx}\phi_k=-x\phi_k+\phi_{k+1}$\lo, so that
\begin{equation}\label{eq3basis}
x\frac{d}{dx}\phi_k =\sum_{j=0}^k{k \choose j-1}(-1)^{k-j}\phi_j=\sum_{j=0}^k \I^{n-j}{k \choose j-1}(-1)^{k-j}\psi_j\lo.
\end{equation}
Finally, we can show that the matrix in terms of which the operator $x\frac{d}{dx}$ is expressed in the basis $\mathcal{B}_2$ is just the matrix $A$.
Indeed, by the very definition of the $\psi_k$'s, see (\ref{basis2}) and eq. (\ref{eq3basis}), we have
\begin{equation}
x\frac{d}{dx}\psi_k = \I^{k-n}x\frac{d}{dx}\phi_k=\sum_{j=0}^k\I^{k-j}{k \choose j-1}\psi_j.
\end{equation}
At the same time $U:\mathcal{B}_1 \rightarrow \mathcal{B}_2$ establishes a change of the corresponding bases, so that the proof is complete.
\end{proof}
Since $R(\lambda)=\e^{-\log\lambda A},$ we obtain from
Lemma \ref{lemma-poisson} the following expression
$$
R(\lambda)=U \e^{-\log\lambda D}U^{-1},
$$
and consequently
\begin{equation}
R(\mu,\lambda)= U \e^{\log\frac{\mu}{\lambda}D}U^{-1}=UK(\mu,\lambda)U^{-1},
\end{equation}
where $K(\mu,\lambda)_{kj}=\left(\frac{\mu}{\lambda}\right)^k \delta_{kj}$\lo. Using the fact that $U^{-1}_{kj}=\I^{n-j}\stfirst{k}{j}$, with $\stfirst{k}{j}$ denoting the signed Stirling numbers of the first kind, one finds the following entries
\begin{equation}\label{entries-of-R}
R(\mu,\lambda)_{kj}=\sum_{l=0}^k \I^{k-j}\left(\frac{\mu}{\lambda}\right)^l{k \brace l}\,\stfirst{l}{j}.
\end{equation}
With the help of (\ref{entries-of-R}) and Lemma \ref{cond_exp_lem} we can calculate the conditional moments
\begin{align}
\EE{x_q^k}[P] &=\sum_{l=0}^k(-\I \mu^{-1})^k(\mathrm{i}\lambda)^l R(\mu,\lambda)_{kl}x_p^l \nonumber\\
&=\sum_{l=0}^k \I^{l-k}\mu^{-k}\lambda^l\left(\sum_{s=0}^k\I^{k-l}\left(\frac{\mu}{\lambda}\right)^s {k \brace s}\,\stfirst{s}{l}\right)x_p^l \nonumber \\
&= \sum_{l=0}^k\sum_{s=0}^k \left(\mu^{s-k}\lambda^{l-s}{k \brace s}\,\stfirst{s}{l}\right)x_p^l\loo.
\end{align}
In particular, the choice $\mu=\abs{q}$ and $\lambda=\abs{p}$ gives
\begin{align}\label{cond-moments}
\EE{x_q^k}[P]
 &= \abs{q}^{-k}\sum_{s=0}^k (\abs{q}/\abs{p})^s{k\brace s}(\abs{p}x_p)_s\lo,
\end{align}
where $(ax)_s\equiv a^{s}(x)_{s,a}$ and $(x)_{s,a}:=x(x-a^{-1})\cdots (x-(s-1)a^{-1})$, $a\in\mathbb{R}$, denotes the falling factorial with parameter $a$. Furthermore, we have used the relation
$$
(ax)_s=\sum_{l=0}^s \stfirst{s}{l}(ax)^l.
$$
It is clear that in (\ref{cond-moments}) all terms with $s<k$ will diverge in the limit $\abs{q} \rightarrow 0.$
In order to obtain finite results in this limit the moments have to be renormalized by adding appropriate counterterms. For this let us define
\begin{align}\label{counterterms}
x^k_{q,\mathrm{ren}}:=
x_q^k+\stfirst{k}{k-1}\,x_q^{k-1}\abs{q}
+\stfirst{k}{k-2}\,x_q^{k-2}\abs{q}^2+\cdots + \stfirst{k}{1}\,x_q \abs{q}^{k-1}.
\end{align}
The next Proposition shows that the renormalized moments just introduced do the job, since the $\abs{q}$-dependent factors disappear.
\begin{prop}\label{renormalized-powers}
The following relation holds for all $k\in \N$:
$$
\EE{x^k_{q,\mathrm{ren}}}[P]=\abs{p}^{-k}(\abs{p}x_p)_k=(x_p)_{k,\abs{p}}\Lo.
$$
\end{prop}
\begin{proof} Let $c_{p,q}:=\abs{q}/\abs{p}.$
Employing formula (\ref{cond-moments}) for each term in the conditional expectation, gives
\begin{align}\label{renormalized-moments}
& \EE{x^k_{q,\mathrm{ren}}}[P] & \nonumber \\
&= \abs{q}^{-k}\biggl(\sum_{s=0}^k c_{p,q}^s\stfirst{k}{s}(\abs{p}x_p)_s
+ \sum_{s=0}^{k-1} c_{p,q}^s\stfirst{k}{k-1}\stsecond{k-1}{s}(\abs{p}x_p)_s\biggr.\nonumber \\
& \quad +\biggl. \sum_{s=0}^{k-2} c_{p,q}^s\stfirst{k}{k-2}\stsecond{k-2}{s}(\abs{p}x_p)_s +\cdots +\sum_{s=0}^{1} c_{p,q}^s\stfirst{k}{1}\stsecond{1}{s}(\abs{p}x_p)_s \biggr)\nonumber \\
&= \abs{q}^{-k}\biggl(c_{p,q}^0(\abs{p}x_p)_0\underbrace{\sum_{l=0}^{k}\stfirst{k}{k-l}\stsecond{k-l}{0}}_{\delta_{k0}}\biggr. \nonumber\\
&\quad +\biggl. c_{p,q}^1(\abs{p}x_p)_1\underbrace{\sum_{l=0}^{k-1}\stfirst{k}{k-l}\stsecond{k-l}{1}}_{\delta_{k1}}
+ \,\cdots + c_{p,q}^{k}(\abs{p}x_p)_k\underbrace{\sum_{l=0}^{0}\stfirst{k}{k-l}\stsecond{k-l}{k}}_{\delta_{kk}}\biggr).
\end{align}
In the last equality we have collected terms with fixed $s$. The Kronecker deltas pop up because the matrices $\left(\stsecond{k}{l}\right)_{k,l\geq 0}$ and $\left(\stfirst{k}{l}\right)_{k,l\geq 0}$ are inverses of each other, see \cite[Section 1.9.1]{stanley2011combinatorics}.
A glance at (\ref{renormalized-moments}) shows that only the last term survives, which proves the assertion.
\end{proof}
Let us now address the case of general monomials.
First we recall that
$
\rho_\lambda(s)=\e^{-\lambda}\lambda^s/s!,
$
where the constant $\alpha$ has been integrated into $\lambda$. This
immediately gives the relation
\begin{equation}\label{shifted-rho}
\rho_\lambda(s+1)=\frac{\lambda}{s+1}\rho_\lambda(s).
\end{equation}
Let $p=p(\mf m)$ and set
\begin{equation}\label{cond_exp_1}
g_n(x_p):=
\EE{x_{q_1}\cdots\, x_{q_n}x_{q_{n+1}}^0}[P]=
\bigl(x(\mf m)^{\diamond n}\bigr)_P
\end{equation}
with pairwise different $q_i$'s obeying $q_i\leq p$ and $p=q_1+\cdots +q_{n+1}$\lo. To evaluate (\ref{cond_exp_1}) one has to integrate w.r.t.\ the measure
\begin{align*}
& \D\nu_{\abs{q_1}}(\abs{q_1}x_{q_1})\,\cdots\, \D\nu_{\abs{q_{n+1}}}(\abs{q_{n+1}}x_{q_{n+1}}) \\
&\quad =\rho_{\abs{q_1}}(t_1)\,\cdots\, \rho_{\abs{q_n}}(t_n)\rho_{\abs{q_{n+1}}}\left(s-\sum_{\ell=1}^n t_\ell \right)\D t_1 \,\cdots\, \D t_n\lo,
\end{align*}
where $s=\abs{p}x_p$ and $t_i=\abs{q_i}x_{q_i}$\Lo.
Moreover, according to Proposition \ref{prop-wick-product}, Wick polynomials do not depend on the volumes $\abs{q_i}$ and we shall assume here and in section \ref{Gauss}, when dealing with the Gaussian case, that $\abs{q_1}=\ldots=\abs{q_{n+1}}=\abs{p}/(n+1)=:\abs{q}$.
Then we have to
calculate the following integral
\begin{align}\label{int-poisson}
g_n(x_p)
&=\frac{1}{\rho_{\abs{p}}(s)}\int
\left(\prod_{\ell=1}^{n} \frac{t_\ell}{\abs{q}}\rho_{\abs{q}}(t_\ell)\right)
\rho_{\abs{q}}\left(s-\sum_{\ell =1}^n t_\ell\right)\D{t_n}\ldots\D{t_1}
\end{align}
For simplicity we shall often write $x$ instead of $x_p$\lo.
First note that the relation
\begin{equation}\label{shifted-g}
g_{n}(x+\abs{p}^{-1})=(x+\abs{p}^{-1})g_{n-1}(x)
\end{equation}
holds.
Indeed, owing to (\ref{shifted-rho}) we have
\begin{align}
& g_n(x+\abs{p}^{-1})\nonumber \\
&\quad =\frac{1}{\rho_{\abs{p}}(s+1)}\int
\left(\prod_{\ell=1}^{n} \frac{t_\ell}{\abs{q}}\rho_{\abs{q}}(t_\ell)\right)
\rho_{\abs{q}}\left(s+1-\sum_{\ell =1}^n t_\ell\right)\D{t_n}\ldots\D{t_1}\lo. \nonumber
\end{align}
Performing the change of variable $-u=1-t_n$, we find
\begin{align}\label{shiftedE}
& g_n(x+\abs{p}^{-1})=\frac{s+1}{\abs{p}\rho_{\abs{p}}(s)}
\int \left(\prod_{\ell=1}^{n-1} \frac{t_\ell}{\abs{q}}\rho_{\abs{q}}(t_\ell)\right)
\frac{u+1}{\abs{q}}\rho_{\abs{q}}(u+1)\,\cdot \nonumber\\
&\quad\cdot\rho_{\abs{q}}\left(s-\sum_{\ell =1}^{n-1} t_\ell-u\right)\D{u}\D{t_{n-1}}\ldots\D{t_1}\lo.
\end{align}
But $(u+1)\rho_{\abs{q}}(u+1)=\abs{q}\rho_{\abs{q}}(u)u^0$ and Lemma \ref{cond_exp_lem} shows that integration w.r.t.\ the variable $u$ just gives the value $1$, which therefore proves relation (\ref{shifted-g}). As $g_0(x)=1=(x)_{0,\abs{p}}$ and
relation (\ref{shifted-g}) equals the recursive identities of the falling factorials $(x)_{n,\abs{p}}$\lo, we have verified that
$g_n(x)=(x)_{n,\abs{p}}$ for all $n\in \N$.
Comparing this result with the statement of Proposition \ref{renormalized-powers}, entails that
\[
\EE{x^n_{q,\mathrm{ren}}}[P]=\bigl(x(\mf m)^{\diamond n}\bigr)_P=(x_p)_{n,\abs{p}}\loo.
\]
This shows not only that in the Poisson case the Wick product is given by falling factorials but also makes explicit its renormalization effect, which here is given by the subtraction of counterterms as in (\ref{counterterms}).

The same technique can be used to obtain a recursive formula for general monomials. Indeed, owing to (\ref{shifted-rho}) the following holds
\begin{align}\label{shiftedfwithpower}
&\left(\frac{t_n+1}{\abs{q}}\right)^k \rho_{\abs{q}}(t_n+1) =\left(\frac{t_n+1}{\abs{q}}\right)^{k-1}\rho_{\abs{q}}(t_n)\nonumber \\
&\quad =\sum_{\ell=0}^{k-1}\binom{k-1}{\ell}\abs{q}^{-(k-\ell-1)}\left(\frac{t_n}{\abs{q}}\right)^\ell \rho_{\abs{q}}(t_n).
\end{align}
Let $g_{\mathbf{k}_n}(x):=\EE{x_{q_1}^{k_1}\,\cdots\, x_{q_n}^{k_n}x_{q_{n+1}}^0}[P],$ where $\mathbf{k}_n=(k_1,\ldots,k_n)\in \N^n$. Repeating the first step of (\ref{shiftedE}) and inserting equality (\ref{shiftedfwithpower}), gives the following recursive formula
\begin{equation}\label{shiftedEmonomial}
g_{\mathbf{k}_n}(x+\abs{p}^{-1})=(x+\abs{p}^{-1})\sum_{\ell=0}^{k_n-1}\binom{k_n-1}{\ell}\abs{q}^{-(k_n-\ell-1)}g_{(\mathbf{k}_{n-1},\ell)}(x),
\end{equation}
with $(\mathbf{k}_{n-1},\ell)=(k_1,\ldots,k_{n-1},\ell)$, or equivalently
\begin{equation*}
g_{\mathbf{k}_n}(x)=x\sum_{\ell=0}^{k_n-1}\binom{k_n-1}{\ell}\abs{q}^{-(k_n-\ell-1)}g_{(\mathbf{k}_{n-1},\ell)}(x-\abs{p}^{-1}).
\end{equation*}
There are several terms in (\ref{shiftedEmonomial}) that will diverge in the ultraviolet limit $\abs{q}\rightarrow 0$. In fact, for each power $k$ only the term with index $\ell=k-1$ is not affected from any divergence. On the other hand we may define a family of renormalized polynomials $g_{\mathbf{k}_n,\mathrm{ren}}$ by subtracting all $\abs{q}$-dependent terms from $g_{\mathbf{k}_n}$. Of course this has to be done for all orders of exponents $(2),(3),\ldots, (k_1),(k_1,2)$ and so on up to $(k_1,\ldots,k_n)$. As can be seen from equality (\ref{shiftedEmonomial}), the resulting polynomials will obey $g_{(0),\mathrm{ren}}(x)=1$ and
\[
g_{\mathbf{k}_n,\mathrm{ren}}(x+\abs{p}^{-1})=(x+\abs{p}^{-1})g_{(\mathbf{k}_{n-1},k_n-1),\mathrm{ren}}(x),
\]
which again is just the recursion relation of falling factorials with parameter $\abs{p}$, so that
\begin{equation}
g_{\mathbf{k}_n,\mathrm{ren}}(x_p)=(x_p)_{\abs{\mathbf{k}_n},\abs{p}}\loo,\quad \text{for all}\; \mathbf{k}_n\in \N^n,
\end{equation}
where $\abs{\mathbf{k}_n}=k_1+\cdots+k_n$. The findings above can be summarized as follows
\begin{prop}
The Wick-products of the Poisson field are given by
\[
\bigl(x(\mf m)^{\diamond\abs{\mathbf{k}_n}}\bigr)_P=
g_{\mathbf{k}_n,\mathrm{ren}}(x_p)=(x_p)_{\abs{\mathbf{k}_n},\abs{p}}\loo,\quad \mathbf{k}_n\in \N^n.
\]
\end{prop}

\subsection{Gauss noise} \label{Gauss}

Now $\hat\rho_\lambda(\xi)=\e^{-\lambda\xi^2/2}$ and the probability density itself reads
\[
\rho_\lambda(s)=\frac{1}{\sqrt{2\pi\lambda}}\e^{-s^2/(2\lambda)}.
\]
Let us first make sure that Hypothesis R applies. For this we shall need the Hermite polynomials $H_k^{1/\lambda}$ of variance $1/\lambda$ whose generating function is given by $\displaystyle{\sum_{k=0}^\infty} H^{1/\lambda}_k(\xi)t^k/k!=\e^{\xi t-t^2/(2\lambda)}$. Also, $\hat{\rho}^{(k)}_\lambda=(-1)^k H^{1/\lambda}_k \hat{\rho}_\lambda$ and the following product formula holds
\[
H_n^{1/\lambda}H_m^{1/\lambda}=\sum_{\ell=0}^m\binom{n}{l}\binom{m}{l}\frac{\ell!}{\lambda^\ell}H_{n+m-2\ell}^{1/\lambda}\LO,
\]
see \cite[Ch.\ 4, section 2.1]{roman2005umbral}.
This leads to $\hat{\rho}_\lambda^{(k)}\hat{\rho}'=(-1)^{k+1}H_k^{1/\lambda}H_1^{1/\lambda}\hat{\rho}_\lambda^2$
and since
\[
H_k^{1/\lambda}H_1^{1/\lambda}=H_{k+1}^{1/\lambda}+\frac{k}{\lambda}H_{k-1}^{1/\lambda}\LO,
\]
it follows that
\[
\hat{\rho}_\lambda^{(k)}\hat{\rho}'=\left(\hat{\rho}_\lambda^{(k+1)}+\frac{k}{\lambda}\hat{\rho}_\lambda^{(k-1)}\right)\hat{\rho}_\lambda,
\]
which in turn confirms Hypothesis R.

We obviously have
\begin{equation}\label{derofGauss}
\rho_\lambda'(s)=-\frac{1}{\lambda}s \rho_\lambda(s).
\end{equation}
Let us define
$h_n(x_p):=\EE{x_{q_1}\cdots\,x_{q_n}x_{q_{n+1}}^0}[P]$. Then, as starting point take the relation
\begin{align}\label{Gauss1}
& F_n(x):=
\int \left(\prod_{\ell=1}^n \frac{t_\ell}{\abs{q}}\rho_{\abs{q}}(t_\ell)\right)\rho_{\abs{q}}\left(\abs px-\sum_{\ell =1}^n t_\ell\right)\D{t_n}\ldots\D{t_1}\nonumber \\
&\quad =\int \left(\prod_{\ell=1}^{n} \frac{t_\ell}{\abs{q}}\rho_{\abs{q}}(t_\ell)\right)\rho_{\abs{q}}(t_{n+1})\rho_{\abs{q}}\left(\abs px-\sum_{\ell =1}^{n+1} t_\ell\right)\D{t_{n+1}}\ldots\D{t_1}\lo.
\end{align}
The definitions just introduced allow us to write $h_n(x)=1/\rho_{\abs{p}}(\abs px) F_n(x)$.
Note that the multiple integral in $(\ref{Gauss1})$ is just a convolution product of $n+1$ rapidly decreasing functions, so that differentiation w.r.t.\ $x$ is the same as differentiation w.r.t.\ $s=\abs px$, or w.r.t.\ any of the $t$ variables, provided the result is multiplied by $\abs{p}$ with the appropriate sign.
For this reason the derivative can be performed w.r.t.\ the variable $t_{n+1}$ applied to $\rho_{\abs{q}}(t_{n+1})$. Since $\rho'_{\abs{q}}(t_{n+1})=-1/\abs{q}t_{n+1}\rho_{\abs{q}}(t_{n+1})$, we find
\[F_n'(x)=-\abs{p} F_{n+1}(x)\]
and therefore
\begin{align}
h_n'(x) &=\left(\frac{1}{\rho_{\abs{p}}(s)}\right)' F_n(x)+\frac{1}{\rho_{\abs{p}}(s)} F_n'(x) \nonumber \\
&= \frac{s}{\rho_{\abs{p}}(s)}F_n(x)-\frac{\abs{p}}{\rho_{\abs{p}}(s)}F_{n+1}(x)\nonumber \\
&= \abs{p}\left(x h_n(x)-h_{n+1}(x)\right), \nonumber
\end{align}
or equivalently
\begin{equation}\label{rec-hermite}
h_{n+1}(x)=x h_n(x)-\abs{p}^{-1}h'_n(x).
\end{equation}
As $h_0(x)=1$ and relation (\ref{rec-hermite}) is precisely that of Hermite polynomials with parameter $1/\abs{p}$, we have verified that $h_n(x)=H_n^{1/\abs{p}}(x)$  for all $n\in \mathbb{N}$.

Finally, let us address expectations of general monomials
$$
h_{\mathbf{k}_n}(x):=\EE{x_{q_1}^{k_1}\cdots\, x_{q_n}^{k_n}x_{q_{n+1}}^0}[P],
$$
where $\mathbf{k}_n=(k_1,\ldots,k_n)\in \N^n.$
We need to calculate the multiple integral
\begin{align}
& F_{\mathbf{k}_n}(x):=\int \left(\prod_{\ell=1}^{n-1} \left(\frac{t_\ell}{\abs{q}}\right)^{k_l}\rho_{\abs{q}}(t_\ell)\right)
\left(\frac{t_n}{\abs{q}}\right)^{k_n}\rho_{\abs{q}}(t_n)
\rho_{\abs{q}}\left(s-\sum_{\ell =1}^{n} t_\ell\right)\D{t_{n}}\ldots\D{t_1},
\end{align}
in terms of which we may write $h_{\mathbf{k}_n}(x)=1/\rho_{\abs{p}}(s)F_{\mathbf{k}_n}(x).$ We now differentiate $h_{\mathbf{k}_n}(x)$ and repeat the argument from above by applying the derivative w.r.t. $t_n$ to the factor $(t_n/\abs{q})^{k_n}\rho_{\abs{q}}(t_n)$. From
\[
\left(\left(\frac{t_n}{\abs{q}}\right)^{k_n}\rho_{\abs{q}}(t_n)\right)'=-\abs{p}\left(\frac{t_n}{\abs{q}}\right)^{k_n+1}\rho_{\abs{q}}(t_n)+\frac{\abs{p}}{\abs{q}}k_n\left(\frac{t_n}{\abs{q}}\right)^{k_n-1}\rho_{\abs{q}}(t_n)
\]
it follows that
\begin{align*}
h_{\mathbf{k}_n}'(x) &=\left(\frac{1}{\rho_{\abs{p}}(s)}\right)' F_{\mathbf{k}_n}(x)+\frac{1}{\rho_{\abs{p}}(s)} F_{\mathbf{k}_n}'(x)\\
&=\frac{\abs{p}}{\rho_{\abs{p}}(s)}\left(xF_{\mathbf{k}_n}(x)-F_{\mathbf{k}_n +1}(x)+\frac{1}{\abs{q}}k_n F_{\mathbf{k}_n -1}(x)\right)
\end{align*}
where $\mathbf{k}_n +1=(k_1,\ldots,k_n,k_n+1)$ and likewise for $\mathbf{k}_n -1$,
so that
\begin{equation}\label{rec-Gauss}
h_{\mathbf{k}_n +1}(x)=x h_{\mathbf{k}_n}(x)-\abs{p}^{-1}h'_{\mathbf{k}_n}(x)+k_n\abs{q}^{-1}h_{\mathbf{k}_n -1}(x).
\end{equation}
Formula (\ref{rec-Gauss}) shows that the polynomials $h_{\mathbf{k}_n}$ {\it would} obey the same recursion relations as the Hermite polynomials if there was not the last term. Due to the factor $\abs{q}^{-1}$ it will also diverge in the limit $\abs{q}\rightarrow 0$. In order to guarantee finiteness in this ultraviolet limit, we define a family of renormalized polynomials $h_{\mathbf{k}_n,\mathrm{ren}}$ by subtracting all $\abs{q}$-dependent terms from $h_{\mathbf{k}_n}$\Loo. Proceeding as in the Poisson case, one gets a new family of polynomials that necessarily obeys the recurrence relations of Hermite polynomials with parameter $1/\abs{p}$, so that
\[
h_{\mathbf{k}_n,\mathrm{ren}}(x)=H^{1/\abs{p}}_{\abs{\mathbf{k}_n}}(x),\quad \text{for all}\;\mathbf{k}_n\in \N^n.
\]
Therefore we may state the following
\begin{prop}
The Wick-products of Gauss fields are given by
\[
\bigl(x(\mf m)^{\diamond\abs{\mathbf{k}_n}}\bigr)_P=
h_{\mathbf{k}_n,\mathrm{ren}}(x_p)=H^{1/\abs{p}}_{\abs{\mathbf{k}_n}}(x_p), \quad
\mathbf{k}_n \in \N^n.
\]
\end{prop}

\section{Quantum field theory}
\subsection{Reflection positivity}

Consider a (possibly signed) measure of the form $a\mu$, for some cylinder density $a\in L^1(X)$. If $a\mu$ is a probability measure (in particular, if it is positive) then it can describe a statistic field theory.
When constructing a quantum field, what matters instead is the \emph{reflection positivity}~\cite{simon2015p, glimm2012quantum} of $a\mu$, for that property enables the reconstruction of a non-commutative observable algebra acting on a Hilbert space by understanding one coordinate as ``imaginary time'' and going back to ``real time''---a trick based on Wick rotation, arguably making Gian-Carlo Wick the single scientist who has most influenced this work. We briefly sketch how the reconstruction theorem works in our setup, entering along the way into the basics of implementing space-time symmetries on the field space.

\begin{defn}
Let $\tau_t:M\rightarrow M$ be a one-parameter group of isometries, thought of as (imaginary) time evolution. We extend it to $\mf M$ by
\[
\tau_t(\mf m) = \set{p\circ\tau_{-t} | p\in\mf m},
\]
thus getting the one-parameter group $\tau^t:\mc O(X)\rightarrow \mc O(X)$ induced by $\tau^t x(\mf m) = x(\tau_t\mf m)$, i.e.\
\begin{align*}
\tau^t 
\left( \int_{\mr M^k} x^{\diamondtimes k}\D\alpha^k \right)
	&= \int_{\mf M^k}  x(\tau_{t}\mf m_1)\diamond\cdots\diamond x(\tau_{t}\mf m_k) \D\alpha(\mf m_1\lo,\dots,\mf m_k) \\
	&= \int_{\mf M^k} x^{\diamondtimes k}\D\beta^k,\quad \beta_{p_1\cdots p_k} = \alpha_{p_1\circ\tau_{t}\cdots p_k\circ\tau_{t}}\LOO.
\end{align*}
Next, given a time slice $M_0\subseteq M$ such that $M = \bigcup \tau_t(M_0)$ with disjoint union, write
\[
M_t = \tau_t(M_0),\quad M_I = \bigcup_{t\in I\subseteq\R} M_t
\]
and define the involution $\theta:M\rightarrow M$ by $\theta \tau_t(m) = \tau_{-t}(m)$, $m\in M_0$\lo. Similarily, this is extended to $\mf M$ by $\theta\mf m = \set{p\circ\theta|p\in\mf m}$ and induces on $\mc O(X)$ the involution 
\[
\left( \int_{\mf M^k} x^{\diamondtimes k}\D\alpha^k \right)^\dagger = \int_{\mf M^k} x^{\diamondtimes k}\D\beta^k,\quad \beta_{p_1\cdots p_k} = \overline{\alpha_{p_1\circ\theta\cdots p_k\circ\theta}}\Lo.
\]
\end{defn}

\begin{defn}
Given $P\in\mc P$ and $I\subseteq\R$ a possibly unbounded interval, let
\[
P_I = \Set{p_{M_I}p | p\in P},\quad X_I = \projlim \Set{X_{P_I} | P\in\mc P}.
\]
We define $\mc O_\mr{poly}\bigl(X_I\bigr)$ just as we defined $\mc O_\mr{poly}(X)$. Formally, its elements can be written as
\[
\sum_{k=0}^n \int_{\mf M_I^k} x^{\diamondtimes k} \D\alpha^k,\quad
\mf M_I = \Set{ \mf m\in\mf M | 
	p_{M_I}\in\mf m}.
\]
Now, using the projection $X\rightarrow X_I$ given at scale $P$ by
\[
x\in X_P\mapsto p_{M_I}x \in X_{P_I}\LO,
\]
we get a canonical inclusion $\mc O_\mr{poly}\bigl(X_I\bigr) \subseteq \mc O_\mr{poly}(X)$. Under this identification, $\mc O_\mr{poly}\bigl(X_I\bigr)$ is the algebra of polynomials $a = \sum_k\int_{\mf M^k} x^{\diamondtimes k}\D\alpha^k$ with kernels $\alpha^k$ satisfying
\[
\alpha_{p_1\cdots p_k}=0 \text{ whenever } \supp p_i\subseteq M\setminus M_I \text{ for some } i\leq k.
\]
We define $\mc O\bigl(X_I\bigr)$ to be the closure under Wick calculus of $\mc O_\mr{poly}\bigl(X_I\bigr)$, i.e.\ the algebra of observables of the form
\[
f^\diamond(a_1\lo,\dots,a_n)\in L^1_\mr{eff}(X),\quad a_i\in\mc O_\mr{poly}\bigl(X_I\bigr)
\]
where $f\in C(\R^n)$ is such that
\[
f(\mc Sa_{1P}\lo,\dots,\mc Sa_{nP}) \in B_\mu(X_P^*),\quad P\in\mc P.
\]
\end{defn}
\begin{rk}
If $I\subseteq J\subseteq\R$, then $\mc O\bigl(X_I\bigr)\subseteq \mc O\bigl(X_J\bigr)$.
\end{rk}

\begin{defn}
The partial algebra of \emph{time-ordered local observables}, denoted by $\cat T\mc O(X)$, is equal, as a vector space, to $\mc O(X)$, but comes equipped with the partial product
\[
ba = b\diamond a,\quad a\in\mc O\bigl(X_I\bigr),\ b\in\mc O\bigl(X_J\bigr)
\]
where $I<J$, i.e.\ $s<t$ for all $s\in I$ and $t\in J$. Note that this is just the pointwise product
\[
(ba)_P(x) = b_P(x)a_P(x),
\]
which is well-defined as a cylinder density by independence.
We emphasize that although the Wick product $b\diamond a$ is always well-defined, the time-ordered product $ba$ makes sense only if $a\in\mc O\bigl(X_I\bigr)$ and $b\in\mc O\bigl(X_J\big)$ for some $I<J$.
Note that $a\mapsto a^\dagger$ is an involution on $\cat T\mc O(X)$, meaning that $ba$ makes sense if, and only if, $a^\dagger b^\dagger$ makes sense and
\[
(ba)^\dagger = a^\dagger b^\dagger.
\]
Note, finally, that $aa^\dagger$ makes sense if $a\in \mc O(X_+)$, where $X_+ = X_{(0,\infty)}$\lo.
\end{defn}
\begin{rk}
Given $a\in\mc O\bigl( X_I \bigr)$, one has $\tau^t(a)\in\mc O\bigl( X_{I+t} \bigr)$ where $I+t = \Set{s+t|s\in I}$. Thus,
\[
\tau^t:\cat T\mc O(X_+)\rightarrow \cat T\mc O(X_+).
\]
for all $t\geq 0$, i.e.\ $\set{\tau^t}_{t\geq 0}$ is a one-parameter semigroup of automorphisms of $\cat T\mc O(X_+)$.
\end{rk}

\begin{defn}
We say that $a\in\mc O(X)$ is \emph{reflection positive} if, and only if,
\[
\EE{bb^\dagger \diamond a} \geq 0,\quad b\in\mc O(X_+),
\]
i.e.\ $\omega(b) = \EE{b\diamond a}$ is a \emph{state} of the partial algebra $\cat T\mc O(X)$.
\end{defn}
\begin{rk}
By independence, $a=1$ is reflection positive. Indeed, for $b\in\mc O(X_+)$ one has that
\[
\EE{bb^\dagger} = \EE{b}\EE{b^\dagger} = \Abs{\EE{b}}^2 \geq 0,
\]
because $\mu$ is $\theta$-invariant.
\end{rk}
\begin{prop} \label{a is reflection positive}
If $a\in\mc O(X)$ is such that $\mc Sa(0)\geq 0$, then $a$ is reflection positive.
\end{prop}
\begin{proof}
Indeed, given $b\in\mc O(X_+)$ one has that
\(
\EE{bb^\dagger\diamond a} = \abs{\mc Sb}^2\mc Sa|_{\xi=0}\lo.
\)
\end{proof}

\begin{theo}\label{reconstruction-theo}
Let $\omega: \cat T\mc O(X)\rightarrow \C$ be a $\tau^t$-invariant state
for which $\tau^t$ is \emph{pointwise weakly continuous,} in the sense that
\[
\lim_{t\rightarrow 0} \omega\bigl(\tau^t(a)b^\dagger\bigr) = \omega(ab^\dagger) 
\]
for all $a,b\in \cat T\mc O(X)$.
Then, there exists a Hilbert space $\mc H$, a partial algebra representation
\[
\pi: \cat T\mc O(X_+)\rightarrow \Set{ A:\dom(A)\subseteq\mc H\rightarrow\mc H },
\]
a self-adjoint operator $H$ on $\mc H$ whose spectrum is bounded from below, and a cyclic, unit vector $\Omega\in\mc H$ such that:
\begin{enumerate}
	\item $\omega(a) = \langle\Omega, \pi(a)\Omega\rangle$ for all $a\in \cat T\mc O(X_+)$.
	\item $\pi\left(\tau^t(a)\right)\Omega = \e^{-t H}\pi(a)\Omega$ for all $a\in \cat T\mc O(X_+)$ and $t\geq 0$. Thus, in particular, $H\Omega = 0$.
\end{enumerate}
\end{theo}
\begin{proof}
This is a GNS-like construction. Consider the vector space
\[
\mc H_0 = \cat T\mc O(X_+)/N, \quad
N = \Set{ a\in \cat T\mc O(X_+) | \omega(aa^\dagger)=0 }.
\]
Using the Cauchy-Schwartz inequality associated to the positivity of $\omega$, we  see that
\[
a\in \cat T\mc O(X_+),\, b\in N \text{ and } ab \text{ makes sense } \Rightarrow ab\in N.
\]
Thus, the (partial) action of $\cat T\mc O(X_+)$ on $\mc H_0$ given by  $\pi(a)(b+N) = ab+N$ on
\[
\dom\pi(a) = \Set{b+N | b\in \cat T\mc O(X_+) \text{ and } ab \text{ makes sense}},
\]
is well defined. The Hilbert space $\mc H$ is the completion of $\mc H_0$ with respect to the positive, sesquilinear form $\langle a+N,b+N\rangle = \omega(ab^\dagger)$. The unit vector is $\Omega = 1+N$ and, by definition, $\omega(a) = \langle\pi(a)\Omega,\Omega\rangle$ for all $a\in\cat T\mc O(X_+)$.

Let us construct $H$. It will be the generator~\cite{davies1980one} of the one-parameter semigroup $\set{T^t}$ given by
\[
T^t(a+N) = \tau^t(a)+N,\quad a\in \cat T\mc O(X_+),
\]
which is well defined because
\begin{align*}
\omega\bigl( \tau^t(a)\tau^t(a)^\dagger \bigr) &= \omega\bigl( \tau^t(a)\tau^{-t}(a^\dagger) \bigr) 	= \omega\bigl( a\tau^{-2t}(a^\dagger) \bigr) \\
	&\leq \omega(aa^\dagger)^{1/2} \omega\bigl( \tau^{2t}(a)\tau^{-2t}(a^\dagger) \bigr)^{1/2} = 0
\end{align*}
whenever $a\in N$.
In order for $H$ to exist
we need $\set{T^t}$ to be:
\begin{itemize}
\item Strongly continuous, or, equivalently~\cite[Corollary 3.1.8]{bratteli2012operator}, weakly continuous, which follows immediately from the hypothesis.
\item Symmetric:
\begin{align*}
\left\langle T^t(a+N), b+N\right\rangle &= \omega\bigl( \tau^{t}(a) b^\dagger \bigr) = \omega\bigl( a \tau^{-t}(b^\dagger) \bigr) = \omega\bigl( a\tau^{t}(b)^\dagger \bigr) \\
	&= \left\langle a+N,T^t(b+N) \right\rangle.
\end{align*}
\end{itemize}
Note that the condition
\(
\pi\left(\tau^t(a)\right)\Omega = \e^{-t H}\pi(a)\Omega
\)
for all $a\in\cat T\mc O(X_+)$ and $t\geq 0$ holds by construction.
\end{proof}
\begin{rk}
A finite inverse temperature version of this theorem can be obtained, along the lines of~\cite{klein1981stochastic}, by using a finite interval as time domain and the theory of one-parameter \emph{local} semigroups~\cite{klein1981construction}.
\end{rk}

\subsection{The free field}



Consider the Gaussian measure $\mu_\mr{free}$ on $L^2(M)$ with covariance
\[
C(\xi,\eta) = \frac12 \Langle\xi, (-\Delta+1)^{-1}\eta\Rangle.
\]
We want to express this as $\exp^\diamond(a)\mu$ for some
\(
a = \int_{\mf M^2} x^{\diamondtimes 2}\D\alpha^2 
\in\mc O_\mr{poly}(X),
\)
where $\mu$ is a Gaussian white noise. In order to do so, start by computing
\begin{align*}
\mc S(x_p) &= \hat\mu^{-1} \EE{x_p\e^{-\I\xi x}}
	= \frac1{\hat\mu} \int x_p\e^{-\I\sum_{q}\abs q\xi_q x_q} \D\mu(x) \\
    &= \e^{\sum_q\abs q\xi_q^2/2} \left( \I\frac1{\abs p}\frac\D{\D\xi_p} \e^{-\sum_q\abs q\xi_q^2/2} \right)
    = -\I\xi_p\loo.
\end{align*}
It follows that the characteristic function of $\exp^\diamond(a)\mu$ is
\begin{align*}
\varphi(\xi)
	&= \hat\mu\,\mc S\exp^\diamond(a) = \hat\mu\exp\mc S(a) \\
	&= \e^{-\frac12\xi^2} \e^{-\sum_{p_1,p_2\in P} \alpha_{p_1p_2} \xi_{p_1}\xi_{p_2}},\quad \xi\in X_P^*\subseteq L^2(M)
\end{align*}
where $\xi^2 = \sum \abs p\xi_p^2$\Lo. 
Thus, we recover the free field if
\[
\sum \alpha_{p_1p_2} \xi_{p_1}\xi_{p_2} = \frac12\Langle \xi,\bigl((-\Delta+1)^{-1}-1\bigr)\xi\Rangle,
\]
i.e.\ $\alpha_{p_1p_2} = \frac12 \Langle p_1\lo,\left((-\Delta+1)^{-1}-1\right)p_2\Rangle$.
\begin{rk}
Observe, however, that if we simply use $\alpha_{p_1p_2} = \frac12 \Langle p_1\lo,(-\Delta+1)^{-1}p_2\Rangle$, then the effective 2-point functions
\[
\frac{\partial^2}{\partial\xi_{p_1}\partial\xi_{p_2}} \varphi(0),\quad p_1\neq p_2,
\]
do not get modified because the factor $\hat\mu(\xi) = \e^{-\frac12\xi^2}$ is diagonal (its logarithm has vanishing crossed derivatives). In other words, the physics of the corresponding quantum theory does not depend on the variance of the reference noise (which we have arbitrarily set to 1 for cells of volume 1), and \emph{the coefficients $\alpha_{p_1p_2}$ are given by the desired propagator.}
\end{rk}

So, at this point we ask ourselves what happens if we change the reference noise in this construction. Take, for instance, a Poisson reference and let  $\mathrm{Poi} := \exp^\diamond(a)\mu$, where $a$ is as above.
As we will show, $\mathrm{Poi}$ has a positive definite characteristic function and therefore qualifies as the Gibbs measure of a statistical mechanical system---but its $n$-point functions (evaluated at pairwise different arguments) are just those of the free field. In other words, again \emph{the reference noise drops out of the quantum model.} This seems to always be the case, and could be interpreted as follows: the reference noise is a choice of regularization of ``Lebesgue measure'' on the space of fields, and the physics of the quantum models constructed using functional integration with respect to it is independent of this choice.

Let us come back to the stochastic positivity of $\mr{Poi}$.
The characteristic function of $\mu$ reads
\[
\hat{\mu}_P(\xi)=\prod_{p\in P}\e^{\abs{p}\left(\e^{-\I\xi_p}-1\right)}
\]
and the $\mc S$-transform becomes
\[
\mc S(x_p)=\e^{-\abs{p}\left(\e^{-\I\xi_p}-1\right)}\EE{x_p \e^{-\I\abs{p}\xi_p x_p}}=\e^{-\abs{p}\left(\e^{-\I\xi_p}-1\right)}\left(\I\frac{1}{\abs{p}}\frac{\D}{\D\xi_p}\e^{\abs{p}\left(\e^{-\I\xi_p}-1\right)}\right)=\e^{-\I\xi_p}.
\]
As characteristic function of $\mathrm{Poi}$ we therefore obtain
\[
\varphi(\xi)
	= \prod_{p\in P}\e^{\abs{p}\left(\e^{-\I\xi_p}-1\right)}\e^{-\sum_{p_1,p_2\in P} \alpha_{p_1p_2} \e^{-\I\xi_{p_1}}\e^{-\I\xi_{p_2}}}
	=: \varphi_{a}(\xi)\varphi_{b}(\xi).
\]
\begin{rk}
We can see that the connected $n$-point functions evaluated at different arguments coincide with those of the free field, as claimed above.
\end{rk}

It is convenient to approach the problem of positive-definiteness for $\mr{Poi}$ from the perspective of Laplace transforms.
Note that here the characteristic function is holomorphic on $\C^{\abs{P}}$ and analytic continuation is for free.
Therefore, consider
\[
\varphi_{L}(\xi)=\varphi_{La}(\xi)\varphi_{Lb}(\xi):=\varphi_a(-\I\xi)\varphi_{b}(-\I\xi),\quad \text{with} \;\xi\in X^\ast_{P,>0} \cong \R_{>0}^{\abs P}\LO.
\]
Not surprisingly, the first factor happens to be the Laplace transform of $\mu_P$ and is therefore positive-definite. For the second factor one needs to see whether it is completely monotone (CM)~\cite[Ch.\ 4, Theorem 6.13]{berg1984harmonic}, more explicitly, whether
$(-1)^{\abs{\gamma}}D^{\gamma}\varphi_{Lb}(\xi)\geq 0$, for all multi-indices $\gamma\in \N^{\abs{P}}$ and $\xi\in X^\ast_{P,>0}$\LO. For $\abs{\gamma}=1$, the condition is simply
\begin{equation}\label{poisson-first-der}
-\frac{\partial}{\partial\xi_q}\varphi_{Lb}(\xi)=-\frac{\partial}{\partial\xi_q}\e^{\sum_{k,l} \alpha_{kl}\e^{-\xi_k-\xi_l}}= 2\sum_l\alpha_{ql}\e^{-\xi_q-\xi_l}\varphi_{Lb}(\xi)\geq 0.
\end{equation}
Since the coefficients $\alpha_{ql}$ are non-negative by hypothesis,  (\ref{poisson-first-der}) is certainly satisfied. As second derivatives we get
\[
\frac{\partial^2}{\partial\xi_q^2}\varphi_{Lb}(\xi)=\left(4\alpha_{qq}\e^{-2\xi_q}+2\sum_{l\neq q}\alpha_{ql}\e^{-\xi_q-\xi_l}\right)\varphi_{Lb}(\xi),
\]
whereas $\frac{\partial^2}{\partial\xi_r\partial\xi_q}\varphi_{Lb}(\xi)$ equals
\[
2\alpha_{qr}\e^{-\xi_q-\xi_r}\varphi_{Lb}(\xi)+4\left(\sum_{l}\alpha_{ql}\e^{-\xi_q-\xi_l}\right)\left(\sum_{l}\alpha_{rl}\e^{-\xi_r-\xi_l}\right)\varphi_{Lb}(\xi)
\]
for $q\neq r$.
Both expressions obey CM. Higher derivatives will contribute just an extra negative sign to each summand that is produced, which is duly compensated by a negative factor from $(-1)^{\abs{\gamma}}.$  We have thus found that the Laplace transform of $\mathrm{Poi}$ fulfills condition CM.

For Gamma fields the same conclusion holds as for Poisson fields, which moreover can be proved along the same lines of reasoning. We shall therefore confine ourselves to provide the ingredients. The characteristic  function and the $\mc S$-transform are now given by
\[
\hat{\mu}_P(\xi)=\prod_{p\in P}(1+\I\xi_p)^{-\abs{p}},\;\;
\text{and}\;\;
\mc S(x_p)=(1+\I\xi_p)^{-1},
\]
respectively. The characteristic function of $\Gamma_G:=\exp^\diamond(a)\mu$ is
\[
\varphi(\xi)
	=\prod_{p\in P}(1+\I\xi_p)^{-\abs{p}}\e^{\sum_{p_1,p_2\in P} \alpha_{p_1p_2}(1+\I\xi_{p_1})^{-1}(1+\I\xi_{p_2})^{-1}},
\]
which is holomorphic on $\C^{\abs{P}}\backslash \{\I\mathbf{1}\}.$
Taking again recourse to the Laplace transform and complete monotonicity, one finds that $\Gamma_G$ represents a Gibbs measure.

\subsection{Models with quartic interaction}\label{phi4-interaction}

Guided by the experience we have gained upon studying the free field, we now consider the possibility of having a self-interacting model with Euclidean (signed) measure
\[
\exp^\diamond(a)\mu,\quad a = \int_{\mr M^2} x^{\diamondtimes}\D\alpha^2 + \int_{\mf M^4} x^{\diamondtimes 4}\D\alpha^4,
\]
where $\alpha_{p_1p_2} = \frac12\Langle p_1,(-\Delta+1)^{-1}p_2\Rangle$ and the coefficients $\alpha_{p_1\cdots p_4}$ specify the connected 4-point function.
In order for $\exp^\diamond(a)$ to exist we need $\mc Sa$ to be bounded from above, and in order for it to be reflection positive we need that $\mc Sa(0)\in\R$.
Besides compatibility, those are the only restrictions on the $\alpha_{p_1\cdots p_4}$'s, and that leaves us with a great deal of freedom. We can take, for instance, 
\[
\alpha_{p_1\cdots p_4} = -\int_M 
	(-\Delta+1)^{-1}p_1(m)\cdots (-\Delta+1)^{-1}p_4(m)\D m,
\]
which can be seen to be well-defined up to $d=8$ using the Sobolev embedding $H^2(M)\hookrightarrow L^p(M)$, $p=2d/(d-4)$. The compatibility conditions
\[
\alpha_{p_1\cdots p_4} = \sum_{q_i\leq p_i} \alpha_{q_1\cdots q_4}
\]
hold by multilinearity and
\[
\mc Sa_P(\xi)
	= \sum \alpha_{p_1\cdots p_4}(\I\xi_{p_1})\cdots (\I\xi_{p_4})
    = -\int_M \left( (-\Delta+1)^{-1}\sum\xi_pp \right)^4\D m \leq 0,
\]
so that $\exp^\diamond(a)$ exists and is reflection positive. This defines a model which is a truncation of the $\phi^4$ field---but makes perfectly good physical sense by itself.

\begin{rk}
If one feels so inclined, higher order Feynman diagrams can be incorporated.
\end{rk}
\begin{rk}
The restriction $d\leq 8$ might be rather easy to remove---the only thing that seems to happen for $d>8$ is that polynomials on the $x_p$'s stop having finite expectations, but \emph{smooth} smearings $x_\rho = \int_M x(m)\rho(m)\D m$ should answer the call of duty.
\end{rk}

\subsection{Outlook}
The model with quartic interaction, as defined above, already provides a very interesting application to the AdS/CFT correspondence. AdS/CFT in its most elementary form amounts to the assertion that a (quantum) field theory on AdS-space gives rise to a conformal field theory on its conformal boundary. One possible description of the Riemannian version of $d$-dimensional AdS-space is given by the manifold
\[
M:=\{m=(z,\boldsymbol{\zeta})=(z,\zeta_1,\ldots,\zeta_{d-1})\in\R^d \,|\, z>0\},
\]
equipped with the metric
\[
ds^2=(dz^2+d\zeta_1^2+\cdots+d\zeta_{d-1}^2)/z^2.
\]
In this parametrization the boundary at infinity, denoted $\partial \mathrm{AdS}$, corresponds to the one-point compactification of the hyperplane $z=0$, which thus can be identified with the $(d-1)$-dimensional unit sphere $\mathbb{S}^{d-1}$. The isometry group of AdS acts by means of conformal transformations on $\partial \mathrm{AdS}$.

Suppose now we had a family of Schwinger-functions $(S_n)_{n\in \N}$ on $\mathrm{AdS}$, obeying the OS-axioms plus the existence of certain scaled limits
\[
S_n^\infty(\boldsymbol{\zeta}_1,\ldots,\boldsymbol{\zeta}_n)=\lim_{(z_1,\ldots,z_n)\rightarrow \mathbf{0}}(z_1\cdots z_n)^\chi S_n((z_1,\boldsymbol{\zeta}_1),\ldots,(z_n,\boldsymbol{\zeta}_n)).
\]
then the boundary Schwinger functions $(S_n^\infty)_{n\in \N}$ themselves satisfy the OS-axioms plus conformal invariance, as shown in \cite{antidesitter2000moschella}.
The real parameter $\chi$ is related to the scaling dimension of the boundary conformal field.
To see how this fits in our setting, let
\begin{equation*}
G(m,m')=(-\Delta +1)^{-1}(m,m')
\end{equation*}
denote the integral kernel of $(-\Delta+1)^{-1}$. In terms of the latter we may write
\begin{equation}\label{alpha-two}
\alpha_{p_1,p_2} = \Big\langle p_1\otimes p_2, G(\cdot,\cdot)\Big\rangle
\end{equation}
and
\begin{equation}\label{alpha-four}
\alpha_{p_1 \cdots p_4} = -\Big\langle p_1\otimes\cdots\otimes p_4, \int_M G(\cdot,m)\cdots G(\cdot,m)\D m\Big\rangle
\end{equation}
and similarly for $\alpha_{p_1,p_2}$. Recall that $\alpha_{p_1,p_2}$ and $\alpha_{p_1 \cdots p_4}$ are just the connected two- and four-point functions of our model, so that in a  more common jargon one would read eq. \ref{alpha-four} by saying that $-\int_M G(\cdot,m)\cdots G(\cdot,m)\D m$ is the connected four-point function evaluated at the ``test'' function $p_1\otimes\cdots\otimes p_4.$ Now evaluating the latter on $\delta_{p_1 z}\otimes\cdots\otimes \delta_{p_4 z},$ a delta function supported on the set $\{m=z\}\cap p_1\cap \ldots \cap p_4$, we may heuristically perform the limit
\begin{equation}\label{bd-four-point}
-\lim_{z\rightarrow 0} z^{4\chi}\int_M G(\delta_{p_1z},m)\cdots G(\delta_{p_4z},m)\D m =-c_1\int_M H(\tilde{p}_1,m)\cdots H(\tilde{p}_4,m)\D m,
\end{equation}
with the $\tilde{p}_i$'s being projections on the boundary corresponding to the $p_i$'s. $H$ is the bulk-boundary propagator encoding the way fluctuations on AdS propagate to the boundary, see~\cite{gottschalk2008ads}. A similar limit for the two-point function
\begin{equation}\label{bd-two-point}
\lim_{z\rightarrow 0}z^{2\chi}G(\delta_{p_1z},\delta_{p_2z})=c_2 G_{\mathrm{bd}}(\tilde{p}_1,\tilde{p}_2), \quad c_1,c_2\in \R,
\end{equation}
gives the boundary two-point function $G_{\mathrm{bd}}$.
Performing these limit operations for all $n$-point functions that can be build from $\alpha_{p_1,p_2}$ and $\alpha_{p_1\cdots p_4}$, we should get a family of $n$-point functions on $\partial \mathrm{AdS}$, that comprise a conformal field on the boundary.

Treatments of the AdS/CFT-correspondence in terms of well-defined and OS-positive functional integrals have given up to now only trivial results, see \cite{gottschalk2009ads1,gottschalk2013ads2}. It would therefore be very interesting to see whether the difficulties encountered so far can be overcome in our framework.

\appendix

\section{The white noise observable algebra}\label{observable algebra}

Here we prove that $B_\mu(X_P^\ast )$ is a subalgebra of $C(X_P^\ast)$ 
assuming, as in \autoref{explicit-calc}, that the support of $\nu_\lambda$ is a fixed additive semigroup $S\subseteq \R$
and that $\D\nu_\lambda(s) = \rho_\lambda(s)\D s$ where $\D s$ is an invariant measure.

It is clear that $B_\mu(X_P^\ast )$ is closed under linear combinations. In order to check that it is closed under multiplication, we need to show that for arbitrary  $\varphi_1,\varphi_2\in B_\mu(X_P^\ast )$ the first two conditions of Definition \ref{fourier-stieltjes} are satisfied. \\
{\it Condition 1.} Since
\[
\varphi_1\varphi_2=\frac{1}{2}\Bigl((\varphi_1+\varphi_2)^2-\varphi_1^2-\varphi_2^2\Bigr),
\]
it is sufficient to verify that $\varphi^2\in B_\mu(X_P^\ast )$, whenever $\varphi\in B_\mu(X_P^\ast )$. But this is true, because  $\varphi^2\hat{\mu}_P^\lambda = \left(\varphi\hat{\mu}_P^{\lambda/2}\right)^2$ and $\varphi\hat{\mu}_P^{\lambda/2}\in B(X_P^\ast).$ One concludes by the fact that $B(X_P^\ast)$ itself is a complex algebra.
\\
{\it Condition 2.} Need to check that $\left(\varphi_1\varphi_2 \hat{\mu}_P^\lambda\right)\raisebox{0.5ex}{$\check{\,}$} \ll \mu_P$ for every $\lambda>0$. Now
\[
\left(\varphi_1\hat{\mu}_P^{\lambda/2}\varphi_2\hat{\mu}_P^{\lambda/2}\right)\raisebox{0.5ex}{$\check{\,}$}
	= \left(\varphi_1\mu_P^{\lambda/2}\right)\raisebox{0.5ex}{$\check{\,}$} \ast \left(\varphi_2\mu_P^{\lambda/2}\right)\raisebox{0.5ex}{$\check{\,}$}
\]
and by hypothesis we know that $\left(\varphi_i\mu_P^{\lambda/2}\right)\raisebox{0.5ex}{$\check{\,}$} \ll \mu_P\Leftrightarrow \left(\varphi_i\mu_P^{\lambda/2}\right)\raisebox{0.5ex}{$\check{\,}$} = f_i\mu_P$ for certain $f_i\in L^1(X_P)$. On the other hand we have for every Borel set $B\subseteq X_P$
\begin{align} \label{absolute}
&f_1\mu_P \ast f_2\mu_P(B) \nonumber \\
& \quad =\int_{X_P}\left(\int_{X_P}1_B(x+y)f_1(y)\D\mu_P(y)\right)f_2(x)\D\mu_P(x).
\end{align}
Recall that $\mu_P$ is the product of the measures $\D\nu_{\abs{p}}(\abs{p}x_p)=\rho_{\abs{p}}(s_p)\D s_p$\lo.
By the invariance of $\D s$, it is clear from (\ref{absolute}) that $f_1\mu_P \ast f_2\mu_P\ll \mu_P$\lo.
Therefore the properties of an algebra hold true.

Let us now prove that $\mc O_{\mathrm{poly}}(X_P)\subseteq \mc O(X_P).$ The $\mc S$-transform is by definition a homomorphism from $\mc O_{\mathrm{poly}}(X_P)$ to $B_\mu(X^\ast_P)$. Therefore
it suffices to verify this for the elementary Wick monomials $\left(x(\mf m)^{\diamond n}\right)_P=x_p^{\diamond n}$. For the latter one has $\mc S\bigl(x_p^{\diamond n}\bigr)=\hat{\mu}^{-n}\mc T\bigl(x_p\bigr)^n.$ The first two conditions of Definition \ref{fourier-stieltjes} now regard the expression
\[
\hat{\mu}^{-n}\mc T\bigl(x_p\bigr)^n \hat{\mu}^\lambda = \left(\hat{\mu}^{-1}\mc T\bigl(x_p\bigr) \hat{\mu}^{\lambda/n}\right)^n,\quad \lambda>0,
\]
whose single factor can be rewritten as
\begin{align*}
&\hat{\mu}^{-1}\mc T(x_p)\hat{\mu}^{\lambda/n} = \hat{\mu}^{\lambda/n-1}\widehat{x_p\mu}
=\hat{\mu}^{\lambda/n-1}\frac{\I}{\abs{p}}\frac{\partial}{\partial\xi_p}\hat{\mu}  \\
	&\quad= \frac{\I}{\abs{p}}(\lambda/n)^{-1} \frac{\partial}{\partial\xi_p}\hat{\mu}^{\lambda/n}
	= (\lambda/n)^{-1} \left( x_p\mu^{*\lambda/n} \right)\raisebox{0.5ex}{$\hat{}$}\ .
\end{align*}
This is the Fourier transform of a complex measure which is absolutely continuous w.r.t.\ $\D s$, implying that $\hat{\mu}^{-1}\mc T(x_p)\in B_\mu(X_P^\ast).$ But $B_\mu(X_P^\ast)$ is an algebra, so that likewise $\hat{\mu}^{-n}\mc T(x_p)^n\in B_\mu(X_P^\ast)$ and the assertion holds.

\bibliographystyle{plain}
\bibliography{biblio}

\end{document}